\documentclass[notitlepage,longbibliography,10pt]{revtex4-1}
\usepackage{amsfonts}
\usepackage{amsmath}
\usepackage{amssymb}
\usepackage{amsthm}
\usepackage{natbib}
\usepackage[mathscr]{eucal}
\usepackage{hyperref}
\usepackage[table]{xcolor}
\usepackage{dsfont}
\usepackage{tikz}
\usetikzlibrary{matrix,arrows}
\usetikzlibrary{external}
\usepackage[normalem]{ulem}
\usepackage{mathtools}
\usepackage{algorithm}
\usepackage{algpseudocode}
\usepackage[margin=1in,papersize={8.5in,11in}]{geometry}

\usepackage{times}

\DeclareMathOperator{\id}{id}
\DeclareMathOperator{\Tr}{Tr}

\DeclareMathOperator{\Span}{Span}

\DeclareMathOperator{\im}{Im}
\DeclareMathOperator{\TOL}{\overleftarrow{\mathcal{T}}}
\DeclareMathOperator{\TOR}{\overrightarrow{\mathcal{T}}}
\newcommand{\identity}{\mathds{1}}
\newcommand{\HH}{\mathcal{H}}  % (general) Hilbert space
\newcommand{\HA}{\HH_{\textsc{a}}}
\newcommand{\HB}{\HH_{\textsc{b}}}
\newcommand{\HAB}{\HA\otimes\HB}

\newcommand{\rhoB}{\rho_{\textsc{b}}}
\newcommand{\TrB}{\Tr_{\textsc{b}}}
\newcommand{\rmd}{{d}} 
\newcommand{\rmT}{\mathrm{T}} 
\newcommand{\LV}[1][]{\mathcal{L}_{#1}}
\newcommand{\LF}[1][]{\mathcal{L}_{F^{\xi}_{#1}}}
\newcommand{\Prj}{\mathcal{P}}
\newcommand{\PrjC}{\mathcal{Q}}
\newcommand{\KO}{\mathcal{K}}  %Koopman operator
\newcommand{\EO}{\mathcal{E}}  %Evolution operator in $\mathfrak{A}$
\newcommand{\PD}[1]{{{#1}_{\ast}}}   %Predual
\newcommand{\BD}[1]{{#1^{\ast}}}   %Banach dual
\newcommand{\BDD}[1]{{#1^{\ast\ast}}}   %Banach double dual

\newcommand{\tikzimageswap}[2]{#1} % Arg #1 is the filename, arg #2 is the tikzpicture 

\newtheorem{theorem}{Theorem}
\newtheorem{lemma}{Lemma}
\newtheorem{corollary}{Corollary}

\theoremstyle{definition}
\newtheorem{definition}{Definition}
\newtheorem{example}{Example}

\newtheorem{remark}{Remark}

\begin{document}
\title{Duality and Conditional Expectations in the Nakajima-Mori-Zwanzig Formulation}
\date{ \today}
\author{Jason M. Dominy}
\email{jdominy@ucsc.edu}
\author{Daniele Venturi}
\email{venturi@ucsc.edu}
\affiliation{Department of Applied Mathematics and Statistics,
University of California, Santa Cruz}
% \thanks{\emph{E-mail addresses}: {\tt jdominy@ucsc.edu}, {\tt venturi@ucsc.edu}}

\newcommand{\jmdstack}[2]{\genfrac{}{}{0pt}{}{#1}{#2}}
\newcommand{\sectionline}{
  \nointerlineskip \vspace{\baselineskip}
  \hspace{\fill}\rule{0.5\linewidth}{.7pt}\hspace{\fill}
  \par\nointerlineskip \vspace{\baselineskip}
}
\newcommand{\red}[1]{\textcolor{red}{#1}}
\newcommand{\ma}[1]{\textcolor{magenta}{#1}}
\newcommand{\blue}[1]{\textcolor{blue}{#1}}
\newcommand{\purp}[1]{\textcolor{violet}{#1}}
\newcommand{\green}[1]{\textcolor{green}{#1}}
\newcommand{\ket}[1]{\left| #1 \right\rangle}
\newcommand{\bra}[1]{\left \langle #1 \right |}
\newcommand{\braket}[2]{\langle #1 | #2 \rangle}
\newcommand{\ketbra}[2]{| #1 \rangle\! \langle #2 |}

\begin{abstract}
We develop a new operator algebraic formulation of the 
Nakajima-Mori-Zwanzig (NMZ) method of projections. The new theory 
is built upon rigorous mathematical foundations, and it can be 
applied to both classical and quantum systems. We show that 
a duality principle between the NMZ formulation in the space of 
observables and in the state space can be established, analogous 
to the Heisenberg and Schr\"odinger pictures in quantum mechanics. 
Based on this duality we prove that, under natural assumptions, 
the projection operators appearing in the NMZ equation must be 
conditional expectations. The proposed formulation is illustrated 
in various examples.
\end{abstract}
\maketitle

\section{Introduction}
High-dimensional stochastic dynamical systems 
arise in many areas of mathematics, natural sciences and 
engineering. Whether it is a physical system being 
studied in a lab, or an equation being solved on a computer, 
the full microscopic state of the system as a point evolving in 
some phase space is often intractable to handle in all its complexity. 

Instead, it is often desirable to attempt to reduce the complexity of the theoretical description by passing from a model of the dynamics of the
full system to a model only of the observables of interest. 
Such observables may be chosen, for example, 
because they represent global macroscopic features 
of the bulk system, as in the derivation of the Boltzmann 
equation of nonequilibrium thermodynamics from microscopic 
descriptions \cite{Villani2002,Snook2006,Cercignani2012}, 
or in the derivation of the dynamics of commutative subalgebras 
of observables in quantum mechanics \cite{Kolomietz2010}.  
The observables may also represent features localized on a subsystem 
of interest, as in the Brownian motion of a particle in a liquid, 
where the master equation governing the position and momentum 
of the particle is derived from first principles 
(Hamiltonian equations of motion of the full system), 
by eliminating the degrees of freedom associated with the 
surrounding liquid \cite{Kampen1986,Chaturvedi1979}.  
In the context of numerical approximation of stochastic 
partial differential equations (SPDEs), the observables may 
be chosen to define a finite-dimensional approximation of the phase space, 
for example a finite set of Fourier-Galerkin 
coefficients \cite{Edwards1964, Herring1966, Montgomery1976}.
Whatever the reason behind this reduction of the set of 
observables, it is often desirable to then attempt to reduce the 
complexity of the theoretical description by passing from a model 
of the dynamics of the full system to a model only of the 
observables of interest. 
For example, we might have a high-dimensional dynamical 
system evolving as $dx/dt = F(x)$, but we are only interested in a 
relatively small number of $\mathbb{C}$-valued observable 
functions $g_{1}(x),\dots,g_{m}(x)$.  The dynamics of this 
lower-dimensional set of observable quantities may be simpler than 
that of the entire system, although the underlying law 
by which such quantities evolve in time is often quite  
complex. Nevertheless, approximation of such law can 
in many cases allow us to avoid performing simulation of 
the full system and solve directly for the quantities of interest.  
If the resulting equation for $\{g_{i}(x)\}$ is low dimensional 
and computable, this provides a means of avoiding the curse 
of dimensionality.

In this paper we study one family of techniques for performing 
such dimensional reduction, namely the Nakajima-Mori-Zwanzig (NMZ) 
method of projections \cite{Nakajima1958,Mori1965,Zwanzig1960,Zwanzig1961}
(see also \cite{Chorin2000, Venturi2016,Venturi2014}).
To this end, we place the NMZ formulation in the context of $C^{*}$-algebras 
of observables, and in so doing, set rigorous foundations 
of this important and widely used technique. More importantly, the operator 
algebraic setting we propose unifies classical and quantum mechanical
formulations.
The method of projections derives its name from the use of a 
projection map from the algebra of observables of the full system, 
to the subalgebra of interest.  In this algebraic context, 
it will naturally emerge that the two common flavors of 
NMZ -- for ``phase space functions'' and for probability density 
functions (PDFs) -- are dual equations for observables and states, 
directly corresponding to the dual Schr\"odinger and Heisenberg 
pictures of quantum mechanics.  Reasoning about information in these 
algebras and desiderata of the NMZ projection will reveal that the 
projection must be a conditional expectation in the operator 
algebraic sense.

The paper is organized as follows.  
We begin in Section \ref{sec:Background} with a quick 
review of $C^{*}$-algebras, their states and homomorphisms, 
as well as the relationship between topological spaces 
and algebras of functions.  We then discuss the relationship 
between classical dynamical systems and observable algebras 
in Section \ref{sec:CompositionTransferOps}, 
deriving from the nonlinear dynamical system the equivalent 
linear dynamics on the observable algebra. 
In Section \ref{sec:GeneralFramework}, the NMZ equation 
is introduced for the reduced dynamics on an observable algebra, 
along with the dual NMZ equation on the states of the algebra.  
We then look more closely at the NMZ projection operator 
in Section \ref{sec:ConditionalExpectations}, finding that, 
under natural assumptions, the projection operator must 
be a conditional expectation.  
While the elements of $C^{*}$-algebras are bounded 
observables, it is common to consider also unbounded observables 
(such as momentum); the incorporation of such affiliated observables 
into the NMZ framework is considered in Section \ref{sec:AffiliatedOperators}. 
In Section \ref{sec:PushForwardState}, we consider the problem 
of ``pushing'' the dynamics from one space to 
another (typically lower dimensional) space using NMZ and 
discuss the application of NMZ to quantum open systems. 
In Section \ref{sec:AnalyticExample} two simple examples
of the NMZ method are carried out analytically.  
Finally, the main results are summarized in Section \ref{sec:Summary}. 
We also include two brief Appendices, in which we discuss technical questions 
related to non-degenerate homeomorphisms and state-preserving maps.

\section{Background}
\label{sec:Background}
In this section we provide a quick review of $C^{*}$-algebras, 
their states and homomorphisms, and relationship between topological 
spaces and algebras of functions. The material in this section is 
well-known and can be found throughout the literature on operator 
algebras and algebraic dynamics.  Standard references for much of this 
material include, e.g., \cite{Kadison1997, Takesaki2002, Bratteli2003,Blackadar2006}.

\subsection{\texorpdfstring{$C^{*}$}{C*}-Algebras of Observables}
% \label{sec:OrderAndStates}
We are interested in developing the NMZ formalism simultaneously 
for dynamical systems
\begin{equation}
\dot{x} = F(x,\xi,t)
\label{eqn:nonautonODE}
\end{equation}
evolving on a (sufficiently) smooth manifold $\mathcal{M}$, 
where $\xi$ represents parameters drawn (perhaps randomly) 
from some parameter space $\Xi$, as well as for quantum 
mechanical systems. By ``manifold'' here we mean any of a 
large class of spaces on which \eqref{eqn:nonautonODE} 
makes sense, including, at a minimum, finite-dimensional manifolds, 
Banach spaces, and more general Banach manifolds. 
In particular, the simple form of \eqref{eqn:nonautonODE} 
can represent many different kinds of initial value problem, 
including ODEs, PDEs, and functional differential equations \cite{Venturi2016a}.  Thinking first of the classical dynamical system 
above, and of $\mathcal{M}$ as a generalized phase space of 
the system, a classical observable will typically 
be a $\mathbb{C}$-valued function on $\mathcal{M}$. 
There are different possible choices for the set of such functions, 
but certain properties may be desirable \cite{Segal1947a, Strocchi2008}. 
For example, if we can observe $f$ and $g$, then we should be able 
to observe $\alpha f + \beta g$ for any $\alpha, \beta\in\mathbb{C}$. 
We should also expect to be able to observe the product $fg$.  
In this way, we should expect the space of observables to form an 
algebra of functions under pointwise addition and multiplication.  
More careful and detailed reasoning \cite{Strocchi2008} about arbitrary 
physical systems (be they classical or quantum) leads to the 
conclusion that the set of observables for any physical system can 
be represented as a $C^{*}$-algebra.  That algebra will generally 
be commutative in the classical case, and noncommutative in the 
quantum case.

A Banach algebra is an algebra $\mathfrak{A}$ over $\mathbb{C}$ with a 
norm making $\mathfrak{A}$ a Banach space, and such 
that $\|xy\|\leq \|x\|\|y\|$ for all $x, y\in\mathfrak{A}$. A $C^{*}$-algebra 
is a Banach algebra $\mathfrak{A}$ with an isometric involution $x\mapsto x^{*}$ 
such that $(xy)^{*} = y^{*}x^{*}$ for all $x,y\in\mathfrak{A}$, and such that $\|x^{*}x\| = \|x\|^{2}$ for all $x\in\mathfrak{A}$.

An important subclass of $C^{*}$-algebras is formed by the von 
Neumann (i.e. $W^{*}$-) algebras, which are unital $C^{*}$-algebras 
closed with respect to the ultraweak topology.  
They can be characterized as those $C^{*}$-algebras which admit a 
Banach space predual \cite{Sakai1956}, i.e. a Banach space whose 
dual space is (isomorphic to) the $C^{*}$-algebra.
Two key commutative algebras of functions  to keep in mind are 
\begin{enumerate}
 \item \texorpdfstring{$\mathfrak{A} = C_{0}(\mathcal{M})$}{A=C0(M)}. 
 This is the algebra of continuous $\mathbb{C}$-valued functions 
 on $\mathcal{M}$ ``vanishing at infinity''.  This means that for 
 any $f\in\mathfrak{A}$ and any $\epsilon > 0$, the 
 set $\{x\in\mathcal{M}\;:\; |f(x)|\geq \epsilon\}$ is compact 
 within $\mathcal{M}$.  The algebra $C_{0}(\mathcal{M})$ is endowed 
 with the $\sup$ norm, i.e. $\|f\| = \sup_{x\in\mathcal{M}}\|f(x)\|$.  It is a unital algebra if and only if $\mathcal{M}$ is compact, in which case the identity is the function that is everywhere equal to one: $\identity(x) = 1$ for all $x\in\mathcal{M}$.
 The dual space $\BD{\mathfrak{A}}$ is isometrically isomorphic 
 to the space $\mathfrak{M}(\mathcal{M})$ of all complex-valued regular 
 Borel measures (i.e. Radon measures) with finite norm 
\begin{equation}
\|\nu\| = \sup_{f\in C_{0}(\mathcal{M})} \frac{\left|\int_{\mathcal{M}}f(x)\,
\rmd\nu(x)\right|}{\|f\|}.
\end{equation}

\item \texorpdfstring{$\mathfrak{A} = L^{\infty}(\mathcal{M},\mu)$}{A=L\{infinity\}(M,m)}.
This is the algebra of (equivalence classes of) essentially 
bounded $\mathbb{C}$-valued measurable functions on $\mathcal{M}$, for some choice of 
localizable \cite{Segal1951} regular Borel measure (i.e. Radon measure) 
$\mu$ on $\mathcal{M}$.  The norm on this algebra is the essential 
supremum $\|f\| = \inf \{r>0\;:\; |f(x)|\leq r \quad\mu\text{-a.e.}\}$.  
The dual space $\BD{\mathfrak{A}}$ is isometrically isomorphic to the 
space of finitely additive localizable complex-valued Borel measures 
absolutely continuous with respect to $\mu$ and with finite norm
\begin{equation}
\|\nu\| = \sup_{f\in L^{\infty}(\mathcal{M},\mu)} 
\frac{\left|\int_{\mathcal{M}}f(x)\,\rmd\nu(x)\right|}{\|f\|}.
\end{equation}
$L^{\infty}(\mathcal{M},\mu)$ is a von Neumann algebra, and, 
as such, admits a unique (up to isometric isomorphism) 
predual $\PD{L^{\infty}(\mathcal{M},\mu)}$, which can be identified 
with $L^{1}(\mathcal{M},\mu)$.  Of course, $L^{1}(\mathcal{M},\mu)$ 
can itself be identified with the space of localizable $\mathbb{C}$-valued 
Borel measures absolutely continuous with respect to $\mu$.

\end{enumerate}

\subsection{Order Structure and States}
\label{sec:OrderAndStates}
A function $g$ in $\mathfrak{A} = C_{0}(\mathcal{M})$ or 
$\mathfrak{A} = L^{\infty}(\mathcal{M},\mu)$ is considered to 
be \emph{positive} if $g(x)\geq 0$ everywhere (or almost everywhere, 
as appropriate).  This will be denoted $g\geq 0$.  More abstractly, 
in a general $C^{*}$-algebra $\mathfrak{A}$, $G\geq 0$ if there 
exist $H\in\mathfrak{A}$ such that $G = H^{*}H$.  The positive 
elements of $\mathfrak{A}$ form a closed convex cone.  An element $\phi$ 
of the Banach dual space $\BD{\mathfrak{A}}$ is positive 
if $\phi(G)\geq 0$ for all $G\geq 0$.  And $\phi\in\BD{\mathfrak{A}}$ 
is a \emph{state} if it is positive and $\|\phi\| = 1$.  Note that, 
in the case that $\mathfrak{A}$ is unital, the normalization 
condition $\|\phi\|=1$ is equivalent to $\phi(\identity) = 1$.  
We will denote the set of states as 
$\mathcal{S}(\mathfrak{A})\subset\BD{\mathfrak{A}}$. 
The equivalence of $\BD{\mathfrak{A}}$ with a Banach 
space of $\mathbb{C}$-valued (at least finitely additive) 
measures on $\mathcal{M}$ implies that $\mathcal{S}(\mathfrak{A})$ 
can be identified with the subset of probability measures 
on $\mathcal{M}$.  Thus for any $\rho\in\mathcal{S}(\mathfrak{A})$ 
and any $G\in\mathfrak{A}$, $\rho(G)$ is the expectation 
value of $G$ over the state (probability measure) $\rho$.  
In other words, if $\hat{\rho}$ is the measure associated 
with the state $\rho$, then $\rho(G)$ can be interpreted 
as 
\begin{equation}
\rho(G) = \int_{\mathcal{M}}G(x)\,\rmd\hat{\rho}(x).
\nonumber
\end{equation}
When $\mathfrak{A}$ is a von Neumann algebra such 
as $L^{\infty}(\mathcal{M},\mu)$, we identify the positive 
cone in the predual $\PD{\mathfrak{A}}$ as follows: 
for any $\phi\in \PD{\mathfrak{A}}$, $\phi\geq 0$ if 
$G(\phi)\geq 0$ for all $G\geq 0$ in $\mathfrak{A}$.  
The set of \emph{normal states} $\mathcal{S}_{N}(\mathfrak{A})$ 
is the set of $\rho\in\PD{\mathfrak{A}}$ such that $\rho\geq 0$ and $\|\rho\|=1$.  
Because of the added complications that finitely additive measures impose, 
it may be advantageous to work with states in the predual when we 
take $\mathfrak{A} = L^{\infty}(\mathcal{M},\mu)$.  In what follows, 
we will use the notation associated with duals (rather than preduals), 
but replacing this with the predual (and normal states) in the case of 
von Neumann algebras is straightforward.

Although $\rho(G)$ is the expectation value of $G$, the state $\rho$ carries
more detailed statistical information about the result of measuring the 
observable $G$.  Indeed, repeated measurements of a normal observable $G$ 
(meaning $G^{*}G = GG^{*}$) against an ensemble of systems in identical 
states does not simply yield the expectation value, but samples from a
probability measure of values.  
Let $\mathfrak{A}$ be any unital $C^{*}$-algebra and $\rho\in\mathcal{S}(A)$ 
be a state; if $\mathfrak{A}$ is not unital, pass to the 
unitization $\tilde{\mathfrak{A}}$ and extend $\rho$ by $\rho(\identity) = 1$ 
to a state of $\tilde{\mathfrak{A}}$.  For any normal $G\in\mathfrak{A}$, 
let $\sigma(G)\subset\mathbb{C}$ (topologized as a compact subset of $\mathbb{C}$) 
be the spectrum of $G$, and $C(\sigma(G))$ the $C^{*}$-algebra of 
continuous $\mathbb{C}$-valued functions on $\sigma(G)$. 
There is a continuous functional calculus, expressible as 
a unital $*$-morphism $\Phi_{G}:C(\sigma(G))\to \mathfrak{A}$, which 
is an isomorphism onto the unital subalgebra of $\mathfrak{A}$ generated 
by $G$ [\onlinecite[Prop. I.4.6]{Takesaki2002}].  Then $\rho_{G} := \Phi_{G}^{*}\rho$ 
is a state of $C(\sigma(G))$, and therefore may be identified through 
the Riesz-Markov theorem with a Radon probability measure on $\sigma(G)$, 
or equivalently with a Radon probability measure $\mu$ on $\mathbb{C}$ 
that is supported on $\sigma(G)$.  This measure describes the probability 
of observing value $\lambda\in\mathbb{C}$ when measuring observable $G$ 
on a system in state $\rho$.
\begin{figure}[t]
	\tikzimageswap{\includegraphics{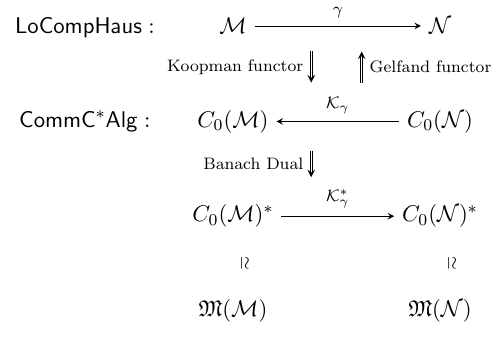}}{
	\tikzsetnextfilename{figure1}
			\begin{tikzpicture}
				\matrix (m) [matrix of math nodes,ampersand replacement=\&, nodes in empty cells,row sep=2.5em,column sep=5em,minimum width=1em]
				{
					\mathsf{LoCompHaus}:\hspace{-4em} \& \mathcal{M} \& \mathcal{N}\\
					\mathsf{CommC^{*}Alg}:\hspace{-4em} \& C_{0}(\mathcal{M}) \& C_{0}(\mathcal{N})\\
					\& \BD{C_{0}(\mathcal{M})}  \& \BD{C_{0}(\mathcal{N})}\\
					\& \mathfrak{M}(\mathcal{M}) \&  \mathfrak{M}(\mathcal{N})\\
			  	};
			  	\path[-stealth,font=\scriptsize]
				(m-1-2) edge node [above] {$\gamma$} (m-1-3)
				(m-2-3) edge node [above] {$\KO_{\gamma}$} (m-2-2)
				(m-3-2) edge node [above] {$\BD{\KO_{\gamma}}$} (m-3-3)
				(m-3-2) edge [-,draw opacity=0] node [above, rotate=270] {$\simeq$} (m-4-2)
				(m-3-3) edge [-,draw opacity=0] node [above, rotate=270] {$\simeq$} (m-4-3)
				([xshift=8ex, yshift=-1ex]m-1-2.south) edge [double] node [left] {Koopman functor} ([xshift=8ex, yshift=2ex]m-2-2.north)
				([xshift=-8ex, yshift=2ex]m-2-3.north) edge [double] node [right] {Gelfand functor} ([xshift=-8ex, yshift=-1ex]m-1-3.south)
				([xshift=8ex, yshift=-1ex]m-2-2.south) edge [double] node [left] {Banach Dual} ([xshift=8ex, yshift=2ex]m-3-2.north);
			\end{tikzpicture}}
	\caption{Sketch of typical spaces and maps associated with the Gelfand representation of locally compact Hausdorff spaces and proper maps.}
	\label{fig:GelfandDuality}
\end{figure}
\subsection{Categories of Interest}
Here we briefly review some categories of topological spaces, 
measure spaces, and algebras that will be relevant to the remainder of the paper.

\begin{itemize}
\item \textsf{LoCompHaus}: locally compact Hausdorff topological spaces 
with proper continuous maps.

\item \textsf{CommC$^{*}$Alg}: commutative (not necessarily unital) $C^{*}$ 
algebras with continuous (i.e. bounded) non-degenerate $^*$-homomorphisms.  
Here $f:\mathfrak{A}\to\mathfrak{B}$ \emph{non-degenerate} 
means $\Span_{\mathbb{C}}\{f(a)b\;:\;a\in\mathfrak{A},b\in\mathfrak{B}\}$ 
is dense in $\mathfrak{B}$ \cite{Brandenburg2007}.   
See also Lemma \ref{lem:NondegenApproxIds} in Appendix \ref{app:nonDegenHomApproxIds} 
for an alternative characterization of non-degenerate morphisms.  
\item \textsf{LocMeas}: Localizable measure spaces \cite{Segal1951} with 
non-singular measurable functions (defined almost everywhere), i.e.  
$f:(X,\Sigma,\mu)\to(Y,\Sigma',\nu)$ is such that $\mu(f^{-1}(A)) = 0$ 
whenever $\nu(A) = 0$.
\item \textsf{CommVNA}: commutative (unital) von Neumann algebras with 
ultraweakly continuous unital $^*$-homomorphisms.
\item \textsf{C$^{*}$Alg}: (not necessarily unital) $C^{*}$ algebras with 
continuous (i.e. bounded) non-degenerate $^*$-homomorphisms.
\item \textsf{vNAlg}:  (unital) von Neumann algebras with ultraweakly 
continuous unital $^*$-homomorphisms.
\end{itemize}

\begin{figure}[t]
	\tikzimageswap{\includegraphics{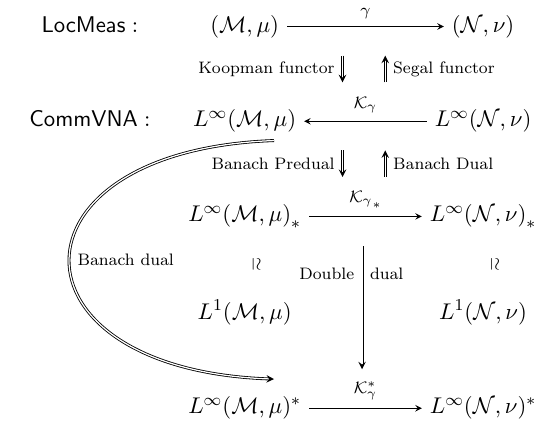}}{
	\tikzsetnextfilename{figure2}
			\begin{tikzpicture}
				\matrix (m) [matrix of math nodes,ampersand replacement=\&, nodes in empty cells,row sep=2.5em,column sep=5em,minimum width=1em]
				{
					\mathsf{LocMeas}:\hspace{-4em} \& (\mathcal{M},\mu) \& (\mathcal{N},\nu)\\
					\mathsf{CommVNA}:\hspace{-4em} \& L^{\infty}(\mathcal{M},\mu) \& L^{\infty}(\mathcal{N},\nu)\\
					\& \PD{L^{\infty}(\mathcal{M},\mu)} \& \PD{L^{\infty}(\mathcal{N},\nu)}\\
					\&  L^{1}(\mathcal{M},\mu) \&  L^{1}(\mathcal{N},\nu)\\
					\& \BD{L^{\infty}(\mathcal{M},\mu)} \& \BD{L^{\infty}(\mathcal{N},\nu)}\\
			  	};
			  	\path[-stealth,font=\scriptsize]
				(m-1-2) edge node [above] {$\gamma$} (m-1-3)
				(m-2-3) edge node [above] {$\KO_{\gamma}$} (m-2-2)
				(m-3-2) edge node [above] {$\PD{\KO_{\gamma}}$} (m-3-3)
				(m-5-2) edge node [above] {$\BD{\KO_{\gamma}}$} (m-5-3)
				(m-3-2) edge [-,draw opacity=0] node [above, rotate=270] {$\simeq$} (m-4-2)
				(m-3-3) edge [-,draw opacity=0] node [above, rotate=270] {$\simeq$} (m-4-3)
				([xshift=10ex, yshift=-1ex]m-1-2.south) edge [double] node [left] {Koopman functor} ([xshift=10ex, yshift=2ex]m-2-2.north)
				([xshift=-10ex, yshift=2ex]m-2-3.north) edge [double] node [right] {Segal functor} ([xshift=-10ex, yshift=-1ex]m-1-3.south)
				([xshift=10ex, yshift=-1ex]m-2-2.south) edge [double] node [left] {Banach Predual} ([xshift=10ex, yshift=2ex]m-3-2.north)
				([xshift=-10ex, yshift=2ex]m-3-3.north) edge [double] node [right] {Banach Dual} ([xshift=-10ex, yshift=-1ex]m-2-3.south)
				([xshift=12ex, yshift=-1ex]m-3-2.south) edge [double] node [left, yshift=1.5em] {Double} node [right, yshift=1.5em] {dual}([xshift=12ex, yshift=2ex]m-5-2.north)
				([xshift=3ex, yshift=0ex]m-2-2.south) edge [double, bend right=88, min distance = 12em, looseness=0] node [right] {Banach dual} ([xshift=3ex, yshift=1ex]m-5-2.north);
			\end{tikzpicture}}
	\caption{Sketch of typical spaces and maps associated with the Segal representation of localizable measure spaces and measurable maps.}
	\label{fig:SegalDuality}
\end{figure}	
When describing the NMZ formalism for a 
dynamical system evolving on a manifold (perhaps infinite-dimensional), 
we will be working largely in the two categories of 
commutative algebras \textsf{CommC$^{*}$Alg} and \textsf{CommVNA}. 
The more general setting -- including the NMZ formalism for quantum mechanics -- 
involves the categories of noncommutative algebras \textsf{C$^{*}$Alg} 
and \text{vNAlg}.

Gelfand duality \cite{Bratteli2003,Brandenburg2007} implies that \textsf{LoCompHaus} is equivalent to \textsf{CommC$^{*}$Alg$^{\mathsf{op}}$} (see Figure \ref{fig:GelfandDuality}).  Thus, for example, any commutative $C^{*}$ algebra is realizable as the algebra $C_{0}(X)$ of continuous functions vanishing at infinity on an essentially unique locally compact Hausdorff (LCH) topological space $X$, and every nondegenerate $C^{*}$-homomorphism between such algebras is realizable as an essentially unique proper map between LCH spaces.  A similar Segal duality \cite{Segal1951} implies that \textsf{LocMeas} is equivalent to \textsf{CommVNA$^{\mathsf{op}}$} (see Figure \ref{fig:SegalDuality}).  These duality theories, which we discuss further in the next subsection, allow for the transformation from nonlinear dynamical systems to completely equivalent linear dynamical systems in Banach spaces.

\subsection{Composition and Transfer Operators}
We will make extensive use of the composition operator 
(i.e., the Koopman operator \cite{Koopman1931}) and 
its adjoint, the transfer operator (i.e., the 
Perron-Frobenius operator) \cite{Beck1995,Sarig2012}.  
Let $\gamma:\mathcal{M}\to\mathcal{N}$ be a continuous map. 
In keeping with our definitions of $\textsf{LoCompHaus}$ 
and $\textsf{LocMeas}$, we require that $\gamma$ is non-singular: 
in $\textsf{LoCompHaus}$ we require that $\gamma$ is 
proper (the pre-image of a compact set is compact), 
and in $\textsf{LocMeas}$ we require 
that $\mu(\gamma^{-1}(A)) = 0$ whenever $\nu(A) = 0$ 
(where $\nu$ is the Borel measure on $\mathcal{N}$). 
Then the composition operator (Koopman) $\KO_{\gamma}$ 
is the $C^{*}$-homomorphism from $\mathfrak{B} = C_{0}(\mathcal{N})$ 
to $\mathfrak{A} = C_{0}(\mathcal{M})$ 
(or from $\mathfrak{B} = L^{\infty}(\mathcal{N},\nu)$ 
to $\mathfrak{A} = L^{\infty}(\mathcal{M},\mu)$), 
defined by 
\begin{equation}
\KO_{\gamma}g = g\circ \gamma.
\label{Koopman}
\end{equation}
As it is a $C^{*}$-homomorphism, $\KO_{\gamma}$ is 
contractive, i.e.,  $\|\KO_{\gamma}\|\leq 1$.  In fact, 
it is easy to see that $\|\KO_{\gamma}\| =1$, since one 
can construct a nontrivial $g\in\mathfrak{B}$ supported in 
the range of $\gamma$, for which we clearly 
have $\|\KO_{\gamma}g\| = \|g\|$ (see Lemma \ref{lem:nonDegenNorm1} 
in Appendix \ref{app:nonDegenHomApproxIds}).

The dual of the composition operator is the transfer operator 
(Perron Frobenius)\cite{Lasota2008} $\BD{\KO_{\gamma}}:\BD{\mathfrak{A}}\to\BD{\mathfrak{B}}$ 
which essentially ``pushes'' states forward along $\gamma$.  
It is straightforward to show that $\BD{\KO_{\gamma}}$ is positive 
and contractive, i.e. $\|\BD{\KO_{\gamma}}\phi\|\leq \|\phi\|$ for 
all $\phi\in\BD{\mathfrak{A}}$.  Moreover, for any $\phi\geq 0$ 
in $\BD{\mathfrak{A}}$,  $\|\BD{\KO_{\gamma}}\phi\| = \|\phi\|$ 
(see Lemma \ref{lem:nonDegenStatePreserv} in Appendix \ref{app:nonDegenHomApproxIds}), 
so that $\BD{\KO_{\gamma}}\mathcal{S}(\mathfrak{A})\subset\mathcal{S}(\mathfrak{B})$. 
When $\mathfrak{A}=C_{0}(\mathcal{M})$ and 
$\mathfrak{B}=C_{0}(\mathcal{N})$, then 
$\BD{\mathfrak{A}}\simeq \mathfrak{M}(\mathcal{M})$ (the space of Radon 
measures on $\mathcal{M}$), 
$\BD{\mathcal{B}}\simeq\mathfrak{M}(\mathcal{N})$, 
and for any measurable $B\subset\mathcal{N}$ and 
$\mu\in\mathfrak{M}(\mathcal{M})$, $(\BD{\KO_{\gamma}}\mu)(B) = 
\mu(\gamma^{-1}(B))$.
In the von Neumann algebra setting, 
when $\PD{\mathfrak{A}}\simeq L^{1}(\mathcal{M},\mu)$ and 
$\PD{\mathfrak{B}}\simeq L^{1}(\mathcal{N},\nu)$, the transfer 
operator $\PD{\KO_{\gamma}}$ may be regarded as a transformation 
between $L^{1}$ functions, defined via the Radon-Nikodym derivative 
\begin{equation}
\big(\PD{\KO_{\gamma}}\rho\big)(y) = \frac{\rmd \eta_{\rho}}{\rmd\nu}(y)
\end{equation}
for all $y\in\mathcal{N}$, where $\eta_{\rho}$ is the push-forward 
measure on $\mathcal{N}$ given by 
\begin{equation}
\eta_{\rho}(A) := \int_{\gamma^{-1}(A)}\rho(x)\rmd\mu(x).
\nonumber
\end{equation}

\subsection{Partial Information}
\label{sec:PartialInformation}
Given the $C^{*}$-algebra of observables $\mathfrak{A}$ on our system, 
we may choose to only observe the system through a subcollection 
$C\subset\mathfrak{A}$.  By only viewing the system through these 
observables in $C$, we obtain only partial information (relative to that 
obtained from using all of the observables in $\mathfrak{A}$).  
However, the set $C$ may not fully represent the set of observables whose 
value we know if we observe using $C$.  For example, if $f,g\in C$, then 
we know the value of $\alpha f + \beta g\in\mathfrak{A}$ for any $a,b\in
\mathbb{C}$. 
We also know the value of $f^{*}$, $g^{*}$, and, at least when $f$ and $g$ 
commute, 
we know the value of $fg$.  So, to represent the set of observables whose 
value we know after measurement, we must expand any mutually commuting set 
of observables $C$ at least to the $*$-algebra containing $C$.  
In this paper (and in keeping with much of the literature in quantum mechanics 
and quantum information) we generally view partial information only 
through the lens of $C^*$-subalgebras of observables.

\section{Nonlinear Dynamical Systems}
\label{sec:CompositionTransferOps}
Consider a nonautonomous dynamical system in the form 
\eqref{eqn:nonautonODE} and assume that the flow $\Phi(t,t_{0})$ 
exists for all $t\geq t_{0}$. Let $\mathfrak{A}$ be a 
commutative $C^{*}$-algebra of $\mathbb{C}$-valued ``observable'' 
functions on a manifold $\mathcal{M}$,  and let $\KO_{\Phi(t,s)}$ 
be the Koopman operator associated with $\Phi(t,s)$.
Clearly, $\KO_{\Phi(t,s)}$ is a $*$-endomorphism acting on $\mathfrak{A}$.

\subsection{Composition and Transfer Operators}
Let $\Phi(t,t_{0})$ be the flow generated by the dynamical system 
\eqref{eqn:nonautonODE}. For any observable $g(x)$ we have
\begin{equation}
 g(x(t))  = \KO_{\Phi(t,t_{0})}g(x_0).
 \label{6}
\end{equation}
By differentiating this equation with respect to $t$ we obtain
\begin{align}
	\frac{\rmd}{\rmd t}\KO_{\Phi(t,t_{0})}g(x_0)
	& = \KO_{\Phi(t,t_{0})}\LF[t]g(x_0),
\end{align}
% \end{subequations}
where $F_{t}^{\xi} = F(\cdot, \xi, t)$ is the vector field from \eqref{eqn:nonautonODE} at fixed $t$ and $\xi$, and $\LF[t]g = \rmd g(F_{t}^{\xi})$ is an ultraweakly densely defined, 
ultraweakly closed, generally unbounded linear operator on $\mathfrak{A}$. 
Therefore,
\begin{align}
	\frac{\rmd}{\rmd t}\KO_{\Phi(t,t_{0})} &= \KO_{\Phi(t,t_{0})}\LF[t], \label{eqn:KoopmanEvolution}
\end{align}
which implies 
\begin{align}
	\KO_{\Phi(t,s)} &= \TOR e^{\int_{s}^{t}\LF[\tau]\,\rmd\tau}.
	\label{9}
% 	g(t) & = \KO_{\Phi(t,0)}g_{0}.
\end{align}
Here $\TOR$ is the time-ordering operator placing later operators to 
the right.  It should be noted that, for time-dependent $\LV[t]$ and 
irreversible flow $\Phi(t,t_{0})$, the observable $g(t)=g(x(t))$ doesn't 
generally obey a simple evolution equation of the form $\rmd g(t)/\rmd t = R_{t}g(t)$. 
This is because there need not exist a time-dependent 
operator $R_{t}$ such that 
\begin{equation}
\frac{\rmd}{\rmd t}g(t) = \KO_{\Phi(t,0)}\LV[t]g_{0} = R_{t}\KO_{\Phi(t,0)}g_{0} = R_{t}g(t).
\end{equation}
Now, consider a fixed state $\rho_{0}\in\mathcal{S}(\mathfrak{A})$. By using 
\eqref{6} and \eqref{9} we have   
% \begin{subequations}
\begin{align}
\rho_{0}(g(t)) = \rho_{0}\left(\TOR e^{\int_{0}^{t}\LF[\tau]\,\rmd\tau}g_{0}\right).
\nonumber
% & = \left(\TOL e^{\int_{0}^{t}\BD{LF[\tau]}\,\rmd\tau}\rho_{0}\right)(g_{0}) 
% = \rho(t)(g_{0})\\
\end{align}
Differentiation with respect to $t$ yields 
\begin{align}
\frac{\rmd}{\rmd t}\big[\rho_{0}(g(t))\big] & = 
\rho_{0}\left(\frac{\rmd}{\rmd t}g(t)\right) \nonumber\\ 
& = \rho_{0}\left(\TOR e^{\int_{0}^{t}\LF[\tau]\,\rmd\tau}\LF[t]g_{0}\right) \notag\\
& = \left(\BD{\LF[t]}\TOL e^{\int_{0}^{t}\BD{\LF[\tau]}\,\rmd\tau}\rho_{0}\right)(g_{0}) \notag\\
& = \left(\BD{\LF[t]}\rho(t)\right)(g_{0}),
\end{align}
i.e., 
\begin{align}
\frac{\rmd}{\rmd t}\rho(t) &= \BD{\LF[t]}\rho(t).
\label{Liouville}
%\\
%\rho(t) & = e^{t\BD{\mathcal{L}}}\rho_{0},
\end{align}
The formal solution to \eqref{Liouville} is 
\begin{equation}
\rho(t)= \BD{\KO_{\Phi(t,s)}}\rho(0),
\end{equation}
where
\begin{equation}
\BD{\KO_{\Phi(t,s)}} = \TOL e^{\int_{0}^{t}\BD{\LF[\tau]}\,\rmd\tau}
\end{equation}
is the transfer operator (Perron-Frobenius) associated to the 
flow map $\Phi(t,s)$.
The Gelfand and Segal dualities (see Fig. \ref{fig:GelfandDuality} 
and Fig. \ref{fig:SegalDuality}) imply that the linear dynamics 
of the composition operator in $\mathcal{B}(\mathfrak{A})$ [or the linear dynamics 
of the transfer operator in $\mathcal{B}(\BD{\mathfrak{A}})$] is completely 
equivalent to the nonlinear dynamics generated 
by \eqref{eqn:nonautonODE} on $\mathcal{M}$.

\section{Operator Algebraic Formulation}
\label{sec:GeneralFramework}
So far we have framed the discussion around dynamical 
systems evolving on manifolds, leading to commutative 
observable algebras $\mathfrak{A}$. However, very little of what 
we will do will depend on the commutativity of $\mathfrak{A}$.  
By and large, if we have any $C^{*}$-algebra $\mathfrak{A}$ 
and a linear evolution operator which is an $*$-endomorphism 
on $\mathfrak{A}$ that evolves as in \eqref{eqn:KoopmanEvolution}, 
it is possible to apply the NMZ formulation to obtain a generalized 
Langevin equation for the reduced dynamics.  
In particular, this applies to quantum mechanics.  
In this section we develop the formalism in this more 
general perspective.

Let $\mathfrak{A}$ be a (not necessarily commutative) $C^{*}$-algebra.  
We will typically let $\BD{\mathfrak{A}}$ denote the Banach space dual, 
and $\mathcal{S}(\mathfrak{A})\subset\BD{\mathfrak{A}}$, the closed 
convex set of positive norm one linear functionals on $\mathfrak{A}$, 
be the set of states on $\mathfrak{A}$.  However, when $\mathfrak{A}$ 
is a von Neumann algebra (i.e. it admits a Banach predual), 
we will abuse the notation to let $\BD{\mathfrak{A}}$ denote 
the Banach space predual, and $\mathcal{S}(\mathfrak{A})\subset\BD{\mathfrak{A}}$,
the closed convex set of positive norm one elements of 
the predual, be the set of (normal) states on $\mathfrak{A}$. 
A (not necessarily bounded) linear operator $\mathcal{L}$ acting 
on $\mathcal{A}$ is called a \emph{$\ast$-derivation} if 
for any $f,g\in\mathfrak{A}$, $\mathcal{L}(fg^{\ast}) = 
(\mathcal{L}f)g^{\ast} + f(\mathcal{L}g)^{\ast}$.

Suppose that we have available a time-dependent family of closed, 
densely defined linear $\ast$-derivations $\{\LV[t]\}_{0\leq t\leq T}$, 
along with a $*$-endomorphism $\EO(t,s)$ on $\mathfrak{A}$,
strongly continuous in both $s$ and $t$, and satisfying for $t\geq s$
\begin{equation}
	\EO(t,s) = \id_{\mathfrak{A}} + \int_{s}^{t}\EO(\tau,s)\LV[\tau]\,\rmd\tau,
\end{equation}
i.e., satisfying, in the sense of Carath{\'e}odory, the differential equation
\begin{equation}
	\frac{\rmd}{\rmd t}\EO(t,s) = \EO(t,s)\LV[t]
\end{equation}
such that $\EO(s,s)$ is the identity morphism $\id_{\mathfrak{A}}$ on $\mathfrak{A}$.  
In the case of von Neumann algebras it suffices for $\{\LV[t]\}$ to be 
weak-*-densely defined and for $\EO(t,s)$ to be weak-* continuous 
in $s$ and $t$.  Then 
\begin{equation}
\EO(t,s) = \TOR e^{\int_{s}^{t}\LV[\tau]\,\rmd\tau}
\label{SemiG}
\end{equation}
is contractive, i.e.,  $\|\EO(t,s)X\|\leq\|X\|$ for 
all $t\geq s$  and all $X\in\mathfrak{A}$, by virtue of 
being a $*$-endomorphism.  This $\EO(t,t_{0})$ serves as the evolution operator for observables in $\mathfrak{A}$, i.e., $\EO(t,t_{0})(G_{0}) = G(t)$.  In the classical case described in Section \ref{sec:CompositionTransferOps}, $\EO(t,t_{0}) = \KO_{\Phi(t,t_{0})}$.

\subsection{Nakajima-Mori-Zwanzig Method of Projections}

We now introduce a projection $\Prj$ on $\mathfrak{A}$ and 
develop from it the equations that comprise the NMZ 
formalism.  The nature and properties of $\Prj$ will be 
discussed in detail in section \ref{sec:ConditionalExpectations}, 
but for now it will suffice to assume only that $\Prj$ is 
a bounded linear operator acting on $\mathfrak{A}$, 
and that $\Prj^{2} = \Prj$.  The NMZ formalism describes 
the evolution of observables initially in the 
image of the $\Prj$.  Because the evolution of observables 
is governed by $\EO(t,t_{0})$ (see equation \eqref{SemiG}), 
we seek an evolution equation for $\EO(t,t_{0})\Prj$.  To this end, 
recall first the well-known Dyson identity: 
if 
\begin{equation}
 Y(t,s) = \TOR e^{\int_{s}^{t}A(\tau)\,\rmd\tau}\quad \textrm{and}\quad 
 Z(t,s) = \TOR e^{\int_{s}^{t}B(\tau)\,\rmd\tau}
\end{equation}
then
\begin{align}
	Y(t,t_{0})-Z(t,t_{0}) &= \int_{t_{0}}^{t}
	\frac{d}{ds}\Big(Y(s,t_{0})Z(t,s)\Big)\,\rmd s \notag\\
	& = \int_{t_{0}}^{t}Y(s,t_{0})(A(s)-B(s))Z(t,s)\,\rmd s.
\end{align}
Applying this to $Y(t,t_{0}) = \EO(t,t_{0})$ (equation \eqref{SemiG}) 
and $Z(t,t_{0}) = \TOR e^{\int_{t_{0}}^{t}\PrjC\LV[\tau]\,\rmd\tau}$ (here 
$\PrjC=1-\Prj$ denotes the complementary projection), we find
\begin{align}
\EO(t,t_{0}) & = \TOR e^{\int_{t_{0}}^{t}\PrjC\LV[\tau]\,\rmd\tau} + 
\int_{t_{0}}^{t}\EO(s,t_{0})\Prj\LV[s]\TOR e^{\int_{s}^{t}\PrjC\LV[\tau]\,\rmd\tau}\,
\rmd s.
\end{align}
A differentiation with respect to time and composition with $\Prj$ yields the 
following generalized Langevin equation for $\EO(t,t_{0})\Prj$
\begin{align}
\frac{\rmd}{\rmd t}\EO(t,t_{0})\Prj & = 
% \EO(t,t_{0})\LV[t]\Prj\\
% & = \EO(t,t_{0})\Prj\LV[t]\Prj + \EO(t,t_{0})\PrjC\LV[t]\Prj\\
\EO(t,t_{0})\Prj\LV[t]\Prj + \TOR e^{\int_{t_{0}}^{t}\PrjC\LV[\tau]\,\rmd\tau}\PrjC\LV[t]\Prj\notag\\
& \quad +  \int_{t_{0}}^{t}\EO(s,t_{0})\Prj\LV[s]\TOR e^{\int_{s}^{t}\PrjC\LV[\tau]\,\rmd\tau}\PrjC\LV[t]\Prj\,\rmd s.
\label{eqn:NMZforCompositionOperator}
\end{align}
Letting $f_{0} = \Prj g_{0}$ represent an observable 
initially in the image of $\Prj$, the NMZ equation describes 
its evolution as
\begin{align}
\frac{\rmd}{\rmd t}f(t) & =  \EO(t,t_{0})\Prj\LV[t]f_{0} + 
\TOR e^{\int_{t_{0}}^{t}\PrjC\LV[\tau]\,\rmd\tau}\PrjC\LV[t]f_{0}\notag\\
& \quad +  \int_{t_{0}}^{t}\EO(s,t_{0})\Prj\LV[s]
\TOR e^{\int_{s}^{t}\PrjC\LV[\tau]\,\rmd\tau}\PrjC\LV[t]f_{0}\,\rmd s.
\label{eqn:NMZforObservables}
\end{align}
The Banach dual of \eqref{eqn:NMZforCompositionOperator} yields 
the NMZ equation for states 
\begin{align}
\frac{\rmd}{\rmd t}\BD{\Prj}& \BD{\EO(t,t_{0})} 
=  \BD{\Prj}\BD{\LV[t]}\BD{\Prj}\BD{\EO(t,t_{0})} + 
\BD{\Prj}\BD{\LV[t]}\TOL e^{\int_{t_{0}}^{t}\BD{\PrjC}
\BD{\LV[\tau]}\,\rmd\tau}\BD{\PrjC} \notag\\
& +  \BD{\Prj}\BD{\LV[t]}\int_{t_{0}}^{t}
\TOL e^{\int_{s}^{t}\BD{\PrjC}\BD{\LV[\tau]}\,\rmd\tau}
\BD{\PrjC}\BD{\LV[s]}\BD{\Prj}\BD{\EO(s,t_{0})}\,\rmd s.
\label{eqn:NMZforTransferOperator}
\end{align}
In this case, the generalized Langevin equation describing  
the evolution of a projected state $\sigma(t) = \BD{\Prj}\rho(t) 
= \BD{\Prj}\BD{\EO(t,t_{0})}\rho_{0}$ is
\begin{align}
\frac{\rmd}{\rmd t}\sigma(t) =& \BD{\Prj}\BD{\LV[t]}\sigma(t) + \BD{\Prj}\BD{\LV[t]}\TOL e^{\int_{t_{0}}^{t}\BD{\PrjC}\BD{\LV[\tau]}\,\rmd\tau}\BD{\PrjC}\rho_{0} \notag\\
& +  \BD{\Prj}\BD{\LV[t]}\int_{t_{0}}^{t}\TOL e^{\int_{s}^{t}\BD{\PrjC}\BD{\LV[\tau]}\,\rmd\tau}\BD{\PrjC}\BD{\LV[s]}\sigma(s)\,\rmd s.
	\label{eqn:NMZforStates}
\end{align}
The NMZ equations \eqref{eqn:NMZforObservables} 
and \eqref{eqn:NMZforStates} describe the exact evolution of 
the reduced observable algebras and states. 
In the context of classical dynamical systems, equations 
\eqref{eqn:NMZforObservables} and 
\eqref{eqn:NMZforStates} describe, respectively, the evolution of 
a phase space function (observable) and the corresponding probability 
density function.  
We would like to emphasize that the duality we just established 
between the NMZ formulations \eqref{eqn:NMZforObservables} and \eqref{eqn:NMZforStates} extends the well-known duality between 
Koopman and Perron-Frobenious operators to reduced observable 
algebras and states.

%\red{$\BD{\Prj}\BD{\EO(t,t_{0})}$ is dual to $\EO(t,t_{0})\Prj$ which describes the evolution of phase-space functions in the image of $\Prj$.}

\subsection{Matrix Form}

Unlike the NMZ evolution equation for states \eqref{eqn:NMZforStates}, 
the evolution equation for observables \eqref{eqn:NMZforObservables} 
involves the explicit application of the full evolution 
operator \eqref{SemiG}.  
In other words, while \eqref{eqn:NMZforCompositionOperator} is 
a generalized Langevin equation for the evolution 
operator $\EO(t,t_{0})\Prj$, \eqref{eqn:NMZforObservables} is 
not explicitly in Langevin form.  However, 
the equation may be put into a Langevin-type form by considering 
a basis for the image of $\Prj$.
To this end, suppose that 
\begin{align}
\{g_{k}\}\subset&\left(\bigcap_{t\in[t_{0},T]}\mathscr{D}(\LV[t])\right)
\cap \nonumber\\
& \left(\bigcap_{t_{0}\leq s\leq t\leq T}
\mathscr{D}\Big(\LV[s]\TOR e^{\int_{s}^{t}\PrjC\LV[\tau]\,
\rmd\tau}\PrjC\LV[t]\Big)\right)\nonumber
\end{align}
is a basis for the image of $\Prj$.   
Let $\{\Omega_{ij}(t)\}$, $\{R_{i}(t)\}$, and $\{K_{ij}(t,s)\}$ be 
the unique functions such that 
% \begin{subequations}
\begin{align}
\Prj\LV[t]g_{i} &= \sum_{j} \Omega_{ij}(t)g_{j}\nonumber\\
R_{i}(t) &= \TOR e^{\int_{t_{0}}^{t}\PrjC\LV[\tau]\,\rmd\tau}
\PrjC\LV[t]g_{i}\nonumber\\
\Prj\LV[s]\TOR e^{\int_{s}^{t}\PrjC\LV[\tau]\,\rmd\tau}\PrjC\LV[t]g_{i} &= 
\sum_{j}K_{ij}(t,s)g_{j},\nonumber
\end{align}
% \end{subequations}
i.e. $\Omega(t)$ and $K(t,s)$ are the matrix representations of $\Prj\LV[t]\Prj$ and $\Prj\LV[s]\TOR e^{\int_{s}^{t}\PrjC\LV[\tau]\,\rmd\tau}\PrjC\LV[t]\Prj$, respectively.
Then
\begin{equation}
\frac{\rmd \vec{g}}{\rmd t}(t) = \Omega(t)\vec{g}(t) +
\vec{R}(t) + \int_{t_{0}}^{t}K(t,s)\vec{g}(s)\,\rmd s.
\end{equation}
This is the generalized Langevin equation for the vector-valued
observable $\vec{g}(t)$.

\section{Conditional Expectations}
\label{sec:ConditionalExpectations}
We now consider a special class of projections 
on $\mathfrak{A}$, i.e., the conditional expectations. 
After a brief introduction to these operators, we argue 
that the projection $\Prj$ used in the NMZ formalism 
should typically be a conditional expectation, and 
discuss the problem of constructing these projections.
In the $C^{*}$-algebra literature, the notion of 
conditional expectation \cite{Umegaki1954} has been be 
developed as a noncommutative generalization of the more 
traditional idea of conditional expectation known from 
probability theory.  It is defined to be a positive contractive 
projection $\Prj$ on a $C^{*}$ algebra $\mathfrak{A}$ with 
image equal to a $C^{*}$-subalgebra $\mathfrak{B}\subset\mathfrak{A}$, 
such that $\Prj(bab') = b\Prj(a)b'$ for 
all $a\in\mathfrak{A}$ and $b,b'\in\mathfrak{B}$.  
This implies also that $\Prj$ is completely 
positive \cite{Nakamura1960}.

Next, we show that under natural assumptions the NMZ 
projection $\Prj$ must be a conditional expectation. 
% Although the NMZ formalism allows us to write 
% equations \eqref{eqn:NMZforCompositionOperator} 
% and \eqref{eqn:NMZforTransferOperator} for arbitrary 
% projections $\Prj$, we want to interpret these as evolution 
% equations for observables and states.  
% This imposes some important restrictions on $\Prj$.  
First, we will want $\Prj$ to project 
onto a $C^{*}$-subalgebra $\mathfrak{B}\subset\mathfrak{A}$.  
This is because the projection represents a restriction to 
partial information about the system, represented 
as a subset of observables, and, as argued in \ref{sec:PartialInformation}, 
such partial information is embodied in the $C^{*}$-subalgebra 
generated by the monitored observables.  
Secondly, keeping the dual pictures in mind, when we
introduce the projection $\Prj$ on $\mathfrak{A}$, we 
want to ensure that $\BD{\Prj}$ preserves states, i.e., 
that $\BD{\Prj}$  maps $\mathcal{S}(\mathfrak{A})$ into itself, 
so that the dual NMZ Langevin equation \eqref{eqn:NMZforStates} 
describes the evolution of states associated with the 
reduced system.  For any $C^{*}$-algebra $\mathfrak{A}$ 
of observables of a system and any state-preserving 
projection $\Prj$ onto a $C^{*}$-subalgebra, $\Prj$ is 
a {\em conditional expectation}, as we show 
in Appendix \ref{app:statePreservingMaps}.

It should be noted that, since $\Prj$ is 
contractive (i.e., $\|\Prj\| = 1$), the complementary 
projection $\PrjC = \identity - \Prj$ is 
bounded: $\|\PrjC\| = \|\identity - \Prj\|\leq\|\identity\| + \|\Prj\| = 2$.  
Indeed, as we will see in Example 2 below, 
the norm of $\PrjC$ can achieve this bound, and therefore $\PrjC$
is not, in general, contractive.

\subsection{Constructing Conditional Expectations}
We aim at constructing explicitly a projection operator 
representing an observable $g(x)$ (see equation \eqref{6}).
In the language of operator algebras this is equivalent to asking 
the following question: How do we find a conditional expectation projecting 
onto a chosen subalgebra $\mathfrak{B}\subset\mathfrak{A}$?  
It should first be noted that, for general 
$\mathfrak{B}\subset\mathfrak{A}$ $C^{*}$-algebras, 
there need not exist a conditional expectation 
$\Prj:\mathfrak{A}\to\mathfrak{B}$. In fact, such 
conditional expectations are 
rare \cite{Takesaki1972,Blackadar2006}. 
However, if $\mathfrak{A}$ admits a faithful, 
tracial state $\tau$ and $\mathfrak{B}\subset\mathfrak{A}$ 
is a nondegenerate $C^{*}$-subalgebra, 
then there exists a {\em unique conditional 
expectation} $\Prj:\mathfrak{A}\to\mathfrak{B}$ such 
that $\tau(AB) = \tau(\Prj(A)B)$ for 
each $A\in\mathfrak{A},B\in\mathfrak{B}$ [\onlinecite[Theorem 7]{Kadison2004}].  
This conditional expectation is ultraweakly 
continuous if $\mathfrak{A}$ and $\mathfrak{B}$ are von Neumann 
algebras and $\tau$ is normal.  
Any faithful state $\rho$ of $\mathfrak{A}$ induces 
an inner product on $\mathfrak{A}$ via 
$\langle x, y\rangle = \rho(x^{*}y)$ \cite{Kadison1962}.  
Then there exists a unique projection from $\mathfrak{A}$ 
onto the closed subspace $\mathfrak{B}$.  
When $\tau$ is a faithful, tracial state, 
this projection is the unique ``$\tau$-preserving'' 
conditional expectation $\Prj$ 
satisfying $\tau(AB) = \tau(\Prj(A)B)$.

On $\mathfrak{A} = L^{\infty}(\mathcal{M},\mu)$, the faithful, 
normal, tracial states are the probability 
measures $\rho$ on $\mathcal{M}$ that are equivalent 
to $\mu$, in the sense that $\mu\ll \rho$ and $\rho\ll \mu$. 
And on $\mathfrak{A} = C_{0}(\mathcal{M})$, the faithful, 
tracial states are the strictly positive Radon probability 
measures, i.e. the Radon probability measures $\mu$ for 
which $\mu(G) > 0$ for all nonempty open sets $G\subset \mathcal{M}$. 
Thus, when $\mathfrak{A}$ is commutative, there always exists a 
conditional expectation onto $\mathfrak{B}\subset\mathfrak{A}$. 
Moreover, for von Neumann $\mathfrak{A}$ and commutative 
von Neumann $\mathfrak{B}\subset \mathfrak{A}$, there 
exists a conditional expectation $\Prj:\mathfrak{A}\to\mathfrak{B}$ 
[\onlinecite[Prop. 6]{Kadison2004}]. Conditional expectations 
can also exist when $\mathfrak{A}$ and $\mathfrak{B}$ are 
both noncommutative.  An example is considered 
in Section \ref{sec:QuantumExample}.  For now, let 
us provide simple examples involving commutative algebras 
of observables.

\begin{example} Consider a three-dimensional dynamical system 
such as the Kraichnan-Orszag system \cite{Orszag1967}, or the 
Lorenz system evolving on the manifold 
$\mathcal{M}=\mathbb{R}^{3}$.
% \begin{align}
% 	\frac{\rmd x_{1}}{\rmd t} =& x_{2}x_{3}\notag\\
% 	\frac{\rmd x_{2}}{\rmd t} =& x_{1}x_{3}\notag\\
% 	\frac{\rmd x_{3}}{\rmd t} =& -2x_{1}x_{2}.\notag
% \end{align}
%
% \begin{subequations}
% \begin{align}
% F&:\mathbb{R}^{3}\to\mathbb{R}^{3}\\
% F(\mathbf{x}) & = [x_{2}x_{3}, x_{1}x_{3}, -2x_{1}x_{2}]\\
% \dot{\mathbf{x}} & = F(\mathbf{x})
% \end{align}
% \end{subequations}}
Let $\mathfrak{A}=L^{\infty}(\mathbb{R}^{3},\lambda)$ where $\lambda$ 
is Lebesgue measure on $\mathbb{R}^{3}$, and 
let $h:\mathbb{R}^{3}\to\mathbb{R}$ (phase space function) be given by 
\begin{equation}
h(x) = e^{-(x_{1}^{2}+x_{2}^{2})},
\label{h}
\end{equation}
where $x_1(t)$ and $x_2(t)$ are the first two phase variables 
of the system.
Then, the subalgebra $\mathfrak{B}$ generated by $h$ 
is the set of all $g\in\mathfrak{A}$ that factor over $h$, i.e., 
functions in the form $g = f\circ h$ for some
$f:\mathbb{R}\to\mathbb{R}$.
In other words, $\mathfrak{B}$ is the subalgebra of $\mathfrak{A}$ comprising 
those functions $g$ that are constant on level sets of $h$. 
A conditional expectation projecting onto this subalgebra may 
be obtained by starting with the faithful, normal, tracial 
state $\tau$ on $\mathbb{R}^{3}$ 
\begin{equation}
\tau(f) = \int_{\mathbb{R}^{3}}f(x)\rmd\mu(x)\nonumber
\end{equation}
given by the Gaussian measure 
\begin{equation} 
\rmd\mu(x) = \frac{1}{(2\pi\sigma^{2})^{\frac{3}{2}}}
e^{-\frac{x_{1}^{2}+x_{2}^{2}+x_{3}^{2}}{2\sigma^{2}}}\rmd\lambda(x).
\end{equation}
The construction of the projection onto the subalgebra  $\mathfrak{B}$ 
generated by \eqref{h} then proceeds as follows. 
Since $g\in\mathfrak{B}$ is constant on level sets of $h$, 
in order to ensure $\tau(fg) = \tau(\Prj(f)g)$ for 
all $f\in\mathfrak{A}$ and all $g\in\mathfrak{B}$, 
we require that $\Prj(f)(x)$ is the mean 
value of $f$ on the level set of $h$ through $x$. 
In other words,
\begin{equation}
(\Prj f)(x) =\begin{cases} 
\displaystyle\int_{\mathbb{R}}\frac{e^{-{y_{3}^{2}}/(2\sigma^{2})}}
{\sigma(2\pi)^{\frac{1}{2}}}f(0,0,y_{3})\,\rmd y_{3}, \quad x_{1} = x_{2} = 0\\
\displaystyle\int_{y\in C(x)}\frac{e^{-y_{3}^{2}(2\sigma^{2})}}
{\sigma(2\pi)^{\frac{3}{2}}\sqrt{x_{1}^{2}+x_{2}^{2}}}\,f(y)\,\rmd y
\end{cases}
\end{equation}
where $C(x)$ is the cylinder $\{y: y_{1}^{2}+y_{2}^{2} = x_{1}^{2}+x_{2}^{2}\}
\subset\mathbb{R}^{3}$ (level set of $h$ 
containing $x$). It is readily verified that this $\Prj$ is 
a projection from $\mathfrak{A}$ to $\mathfrak{B}$ and 
that $\tau(fg)= \tau(\Prj(f)g)$, as desired.  
Of course, this projection works for any 
dynamical system on $\mathbb{R}^{3}$ where we're 
measuring \eqref{h}.  
\end{example}

\begin{example}
Let $(S^{1},\lambda)$ be the circle with 
normalized Haar measure $\lambda$, represented for example 
as $S^{1} = \{e^{2i\pi s}\;:\;s\in[0,1)\}$ with $\lambda$ Lebesgue 
measure on $[0,1)$.  Let $(\mathcal{M},\mu) = T^{2}$ be the two-dimensional 
torus with normalized Haar measure, i.e., $(\mathcal{M},\mu) = 
(S^{1},\lambda)\times (S^{1},\lambda)$.  Let $\mathfrak{A} = L^{\infty}(\mathcal{M},\mu) 
= L^{\infty}(S_{1},\lambda)^{\otimes 2}$.  
Now consider the observable function $g\in\mathfrak{A}$ given by 
\begin{equation}
g\left(e^{2i\pi s_{1}}, e^{2i\pi s_{2}}\right) = e^{2i\pi s_{1}}.
\end{equation}
The von Neumann subalgebra $\mathfrak{B}\subset\mathfrak{A}$ 
generated by $g$ is the subalgebra of functions that factor 
over $g$.  In other words, $\mathfrak{B}$ is the subalgebra 
of functions $h\in\mathfrak{A}$ that are constant on level 
sets of $g$
\begin{equation}
h\left(e^{2i\pi s_{1}},e^{2i\pi s_{2}}\right) = 
h\left(e^{2i\pi s_{1}},1\right)\quad \forall s_{2}\in[0,1).
\end{equation}
A conditional expectation $\Prj$ onto $\mathfrak{B}$ is given by
\begin{equation}
\big(\Prj f\big)(e^{2i\pi s_{1}},e^{2i\pi s_{2}}) = \int_{0}^{1}f\big(e^{2i\pi s_{1}},e^{2i\pi r}\big)\,\rmd r.
\end{equation}
Next we show that the complementary projection $\PrjC=1-\Prj$ is not a contraction 
in this case. To this end, let 
\begin{equation}
f_{n}(e^{2i\pi s_{1}}, e^{2i\pi s_{1}}) = -1 + 2e^{-n^{2}\frac{(s_{2}-1/2)^{2}}{2}}
\end{equation}
so that $\|f_{n}\|_{\infty} = 1$, $\Prj f_{n}\equiv v_{n}$ is a constant function with value 
\begin{equation}
v_{n} = -1 + 2\int_{0}^{1}e^{-n^{2}\frac{(r-1/2)^{2}}{2}}\,\rmd r \searrow -1
\end{equation} 
and $\|\PrjC f_{n}\| = \|f_{n}-v_{n}\identity\| \nearrow 2$, so that
\begin{equation}
\|\PrjC\| = 2.
\end{equation}
\end{example}

\subsection{Unbounded Observables}
\label{sec:AffiliatedOperators}
By their nature, $C^{*}$-algebras comprise \emph{bounded} observables of the system.  In many cases, however, one is interested in unbounded observables.  For example, in studying physical particle systems, we are often interested in the positions and momenta of the particles.  The theory of $C^{*}$-algebras and von Neumann algebras have been extended (in two independent ways) to incorporate certain classes of well-behaved unbounded operators.  In both cases, these operators are called \emph{affiliated} operators, and the underlying idea is to identify those unbounded operators that may in some sense be approximated by the (bounded) elements of the algebra.

\begin{enumerate}
\item {\em von Neumann Algebras}.
When $\mathfrak{A}$ is a von Neumann algebra acting 
on a Hilbert space $\HH$, a closed, densely-defined 
operator $T$ is affiliated with $\mathfrak{A}$ whenever 
$U^{*}TU = T$ for all unitary operators $U\in\mathcal{B}(\HH)$ 
that commute with $\mathfrak{A}$ \cite{Murray1936, Kadison1997}.  
For example, in the case $\mathfrak{A} \simeq L^{\infty}(X,\sigma,\mu)$ 
of a commutative von Neumann algebra isomorphic (via Segal duality) to 
the algebra of essentially bounded functions on a localizeable 
measure space, the affiliated operators $\mathfrak{A}^{\eta}$ are 
represented by the set a of all measurable $\mathbb{C}$-valued 
functions on $(X,\sigma,\mu)$.  Given a von Neumann 
algebra $\mathfrak{A}$ and a normal affiliated operator $T$, there 
is a minimal von Neumann subalgebra $\mathfrak{B}\subset\mathfrak{A}$ 
such that every $T$ is affiliated with $\mathfrak{B}$.  
This $\mathfrak{B}$ is the Abelian subalgebra generated 
by $T$ [\onlinecite[Thm. 5.6.18]{Kadison1997}].  With respect to the 
polar decomposition $T = U|T|$, $\mathfrak{B}$ contains $U$ as 
well as all spectral projections of $|T|$ [\onlinecite[Lemma 2.5.8]{Bratteli2003}].

\item {\em $C^{*}$-Algebras}.
When $\mathfrak{A}$ is a $C^{*}$-algebra, a densely-defined 
operator $T$ acting on $\mathfrak{A}$ is affiliated 
with $\mathfrak{A}$ when $T$ admits a $T^{*}$ such 
that $a^{*}T(b) = [T^{*}(a)]^{*}b$ for all 
$b\in\mathscr{D}(T)\subset\mathfrak{A}$ and all $a$ in 
a dense subset of $\mathfrak{A}$, and when 
$\id_{\mathfrak{A}} + T^{*}T$ has a dense range in $\mathfrak{A}$ 
\cite{Woronowicz1991, Woronowicz1992, Lance1995}. 
For example, when $\mathfrak{A}\simeq C_{0}(X)$ for some 
locally compact Hausdorff space $X$, the affiliated operators 
are represented by the set $C(X)$ of all continuous $\mathbb{C}$-valued 
functions on $X$ \cite{Woronowicz1991}.  If $\mathfrak{A}$ is unital, 
then the set $\mathfrak{A}^{\eta}$ of affiliated operators may be 
identified with $\mathfrak{A}$ itself, which is analogous to 
the statement that every continuous function on a compact 
space is bounded \cite{Woronowicz1991}.  

\end{enumerate}

As in the case of measuring a bounded normal operator 
(see Section \ref{sec:OrderAndStates}), there is a continuous 
functional calculus for any normal affiliated operator $T\in\mathfrak{A}^{\eta}$, 
i.e. a injective nondegenerate $*$-morphism 
$\Phi_{T}:C_{0}(\sigma(T))\to M(\mathfrak{A})$, where $M(\mathfrak{A})$ 
is the multiplier algebra of $\mathfrak{A}$ \cite{Woronowicz1991, Woronowicz1992}. 
Any state $\rho\in\mathfrak{A}$ extends uniquely to a state 
of $M(\mathfrak{A})$ and restricts to a state of $C_{0}(\sigma(T))$, 
namely $\rho_{T} := \BD{\Phi_{T}}\rho$.  As before, the Riesz-Markov theorem 
then yields a probability measure on $\mathbb{C}$, supported 
on $\sigma(T)$, which is interpreted as the probability of the possible 
outcomes of measurement of $T$ on a system in state $\rho$.  
When $\mathfrak{A}$ is unital, $M(\mathfrak{A}) = \mathfrak{A}$, and 
the image of $\Phi_{T}$ is the subalgebra of $\mathfrak{A}$ generated 
by $T$.

\begin{example} Consider a nonlinear dynamical system evolving 
on $\mathcal{M} = \mathbb{R}^{n}$, for example the semi-discrete 
form of an initial/boundary value problem for a PDE. Let, 
$\mathfrak{A} = C_{0}(\mathbb{R}^{n})$ and   
$\mathfrak{B}$ be the subalgebra generated by the observable 
\begin{equation}
h(x) = \sum_{i=1}^{n}a_{i}x_{i},
\end{equation}
for some fixed $a=(a_1,...,a_n)\in\mathbb{C}^{n}$. Note that $h(x)$ may 
represent the series expansion of the solution to the aforementioned 
PDE. Although $h\notin\mathfrak{A}$ (it doesn't vanish at infinity), 
we can still represent the partial information embodied 
in $h$ by the subalgebra $\mathfrak{B}\subset\mathfrak{A}$ of 
functions $g\in\mathfrak{A}$ that factor over $h$, i.e. 
for which there exists $r$ such that $g = r\circ h$.  
Thus, $\mathfrak{B}$ comprises functions in $C_{0}(\mathcal{M})$ 
that are constant on level sets of $h$. A conditional expectation 
onto this $\mathfrak{B}$ is given by
\begin{equation}
(\Prj f)(x) = e^{\frac{|\langle x,a\rangle|^{2}}
{\|a\|^{2}}}\int_{C(x)}\frac{e^{-\frac{\|y\|^{2}}{2\sigma^{2}}}}
{(2\pi\sigma^{2})^{\frac{n-1}{2}}}f(y)\,\rmd y,
\end{equation}
where $C(x) = \{y: h(y) = h(x)\}$.
\end{example}

\section{Dimensional Reduction}
\label{sec:PushForwardState}
In many cases, the reduction to a coarser algebra 
of observables involves pushing the problem from the 
original phase space $\mathcal{M}$ to a new (typically smaller, 
lower dimensional) space $\mathcal{N}$ via a 
continuous (and appropriately nonsingular) 
map $\gamma:\mathcal{M}\to\mathcal{N}$.  In other words, 
rather than worrying about how a state 
$\rho_{0}\in\mathcal{S}(\mathfrak{A})$ evolves, we are 
interested only in how the push-forward state 
$\sigma_{0} := \BD{\KO_{\gamma}}\rho_{0}\in\mathcal{S}(\mathfrak{B})$ 
evolves, where $\KO_{\gamma}:\mathfrak{B}\to\mathfrak{A}$ is the 
Koopman homomorphism from the appropriate observable 
algebra $\mathfrak{B}$ on $\mathcal{N}$ to the original 
algebra $\mathfrak{A}$ of observables on $\mathfrak{A}$ 
and $\BD{\KO_{\gamma}}$ is the corresponding transfer 
(i.e., Perron-Frobenius) operator pushing states forward 
from $\mathcal{M}$ to $\mathcal{N}$.  A similar situation can 
arise in the case of noncommutative algebras, for example 
when $\KO:\mathfrak{B}\to\mathfrak{A}$ is an embedding identifying 
the subalgebra of observables localized on a quantum subsystem 
of interest.  To keep the discussion general, we will assume in 
this section that $\EO(t,s)$ is a strongly continuous 
(or weak-* continuous, in the case of von Neumann algebras) family 
of *-endomorphisms generated by derivations $\LV[t]$ 
on $\mathfrak{A}$ and that $\KO:\mathfrak{B}\to\mathfrak{A}$ is 
a nondegenerate $*$-homomorphism.  We then seek an appropriate 
evolution equation for the reduced state $\sigma(t) = \BD{\KO}\rho(t)$.  
To this end, consider $g_{0}\in\mathfrak{B}$. We have 
\begin{align}
	[\KO g](t)  = & \EO(t,0)\KO g_{0} \nonumber \\ 
	= &\TOR e^{\int_{0} ^{t}\LV[\tau]\,\rmd \tau}(g_{0}\circ\gamma),
\end{align}
and
\begin{align}
	\frac{\rmd}{\rmd t}[\KO g](t) & = \TOR e^{\int_{0}^{t}\LV[\tau]\,\rmd \tau}\LV[t]\KO g_{0} \notag\\
	& = \TOR e^{\int_{0}^{t}\LV[\tau]\,\rmd\tau}\LV[t](g_{0}\circ\gamma).
\end{align}
Then for any state $\rho_{0}\in\mathcal{S}(\mathfrak{A})$
%\begin{subequations}
\begin{align}
	\rho_{0}([\KO g](t)) & = \rho_{0}(e^{t\mathcal{L}}\KO g_{0}) = \rho(t)(\KO g_{0}) = \left[\BD{\KO }\rho(t)\right](g_{0}).
\end{align}
%\end{subequations}
%\begin{subequations}
Therefore, as expected,
\begin{align}
	\frac{\rmd}{\rmd t}\BD{\KO }\rho(t) & = \BD{\KO }\BD{\LV[t]}\rho(t). 
	\label{eqn:rawPushForwardEvolution}
\end{align}
%\end{subequations}
%
This equation is still not 
a reduced-order equation, since the right hand side is 
not in terms of $\sigma_{t}:=\BD{\KO }\rho_{t}$. 
However, we can use the NMZ projection operator method to derive 
the reduced-order equation we are interested in.  To this end, 
let us assume that $\KO:\mathfrak{B}\to\mathfrak{A}$ is injective; 
if not, one can typically replace $\mathfrak{B}$ by $\mathfrak{B}/\ker\KO$ 
and replace $\KO$ by $\tilde{\KO}:\mathfrak{B}/\ker\KO\to\mathfrak{A}$. 
In the case $\KO = \KO_{\gamma}$ is the Koopman morphism of 
a map $\gamma:\mathcal{M}\to\mathcal{N}$, this amounts to 
replacing $\mathcal{N}$ with $\tilde{\mathcal{N}} = \gamma(\mathcal{M})$, i.e. 
the image of $\gamma$, and letting $\mathfrak{B}$ be the appropriate algebra 
of observables on $\tilde{\mathcal{N}}$, say $C_{0}(\tilde{\mathcal{N}})$.  
Since $\KO$ is injective *-morphism, $\KO:\mathfrak{B}\to\mathfrak{A}$ is 
an embedding of $\mathfrak{B}$ into $\mathfrak{A}$.  
In other words, $\hat{\mathfrak{B}}:=\im\KO\subset \mathfrak{A}$ is 
isomorphic to $\mathfrak{B}$ via $\KO$.  Suppose that 
$\Prj:\mathfrak{A}\to\mathfrak{A}$ is a conditional expectation 
onto $\hat{\mathfrak{B}}$.  Then $\Prj$ can be decomposed as the 
composition of two positive contractions: $\Prj = \KO\circ\pi$, where 
$\pi:\mathfrak{A}\to\mathfrak{B}$ may be viewed as the projection 
$\Prj$ onto $\hat{\mathfrak{B}}$, followed by identification of 
$\hat{\mathfrak{B}}$ with $\mathfrak{B}$.  Moreover, it is clear 
that $\pi\circ\KO$ is the identity map on $\mathfrak{B}$, so that 
$\Prj\circ\KO = \KO$, and $\pi\circ\Prj = \pi$.
Using $\Prj$, we get the standard NMZ evolution equation 
for $\BD{\Prj}\rho(t)$:
\begin{align}
\frac{\rmd}{\rmd t}\BD{\Prj}\rho(t) & = \BD{\Prj}\BD{\LV[t]}\BD{\Prj}\rho(t) 
+ \BD{\Prj}\BD{\LV[t]} \TOL e^{\int_{0}^{t}\BD{\PrjC}\BD{\LV[\tau]}\,
\rmd\tau}\BD{\PrjC}\rho_{0} \notag\\
&  + \BD{\Prj}\BD{\LV[t]}\int_{0}^{t}\TOL e^{\int_{s}^{t}\BD{\PrjC}\BD{\LV[\tau]}
\,\rmd\tau}\BD{\PrjC}\BD{\LV[s]}\BD{\Prj} \rho(s)\,\rmd s.
\end{align}
Now, replacing $\BD{\Prj}$ with $\BD{\Prj} = \BD{\pi}\BD{\KO}$, acting 
on the left with $\BD{\KO}$, and using the fact that 
$\BD{\KO}\BD{\Prj} = \BD{\KO}$, we get the desired Langevin 
equation for $\sigma(t) = \BD{\KO}\rho(t)$:
% \begin{subequations}
\begin{align}
\frac{\rmd}{\rmd t}\BD{\KO }\rho(t) & = 
\BD{\KO }\BD{\LV[t]}\BD{\pi}\BD{\KO}\rho(t) + 
\BD{\KO }\BD{\LV[t]} \TOL 
e^{\int_{0}^{t}\BD{\PrjC}\BD{\LV[\tau]}\,\rmd\tau}\BD{\PrjC}\rho_{0} \notag\\
&  + \BD{\KO }\BD{\LV[t]}\int_{0}^{t}\TOL e^{\int_{s}^{t}\BD{\PrjC}\BD{\LV[\tau]}\,
\rmd\tau}\BD{\PrjC}\BD{\LV[s]}\BD{\pi}\BD{\KO} \rho(s)\,\rmd s,\\
\frac{\rmd}{\rmd t}\sigma(t) & = \BD{\KO }\BD{\LV[t]}\BD{\pi}\sigma(t) + 
\BD{\KO }\BD{\LV[t]} \TOL e^{\int_{0}^{t}\BD{\PrjC}\BD{\LV[\tau]}\,\rmd\tau}
\BD{\PrjC}\rho_{0} \notag\\
&  + \BD{\KO }\BD{\LV[t]}\int_{0}^{t}\TOL e^{\int_{s}^{t}\BD{\PrjC}\BD{\LV[\tau]}\,
\rmd\tau}\BD{\PrjC}\BD{\LV[s]}\BD{\pi}\sigma(s)\,\rmd s.\label{eqn:PushForwardStateNMZ}
\end{align}
% \end{subequations}
If $\rho_{0}\in\mathcal{S}(\mathfrak{A})$ is supported on $\hat{\mathfrak{B}}$, then $\rho_{0} = \BD{\pi}\sigma_{0}$ for some $\sigma_{0}\in\mathcal{S}(\mathfrak{B})$.  Then, since $\pi\Prj = \pi$, it follows that $\pi\PrjC = 0$ and $\BD{\PrjC}\BD{\pi} = 0$, so that $\BD{\PrjC}\rho_{0} = \BD{\PrjC}\BD{\pi}\sigma_{0} = 0$ and the ``random noise'' term in the Langevin equation vanishes, leaving
\begin{align}
\frac{\rmd}{\rmd t}\sigma(t) & = \BD{\KO }\BD{\LV[t]}\BD{\pi}\sigma(t)  \notag\\
&  + \BD{\KO }\BD{\LV[t]}\int_{0}^{t}\TOL e^{\int_{s}^{t}\BD{\PrjC}\BD{\LV[\tau]}\,
\rmd\tau}\BD{\PrjC}\BD{\LV[s]}\BD{\pi}\sigma(s)\,\rmd s.
\label{eqn:PushForwardStateNMZReduced}
\end{align}
As we will see in the next example, this equation takes a particularly 
simple form if the dynamics is on a manifold $\mathcal{M}$ with a tensor 
product structure.

\begin{example} Let $(\mathcal{N},\nu)$ 
and $(\mathcal{R},\eta)$  manifolds with Borel measures 
$\nu$ and $\eta$, respectively, and let 
$(\mathcal{M},\mu) = (\mathcal{N},\nu)\times (\mathcal{R},\eta)$, 
i.e. $\mathcal{M} = \mathcal{N}\times \mathcal{R}$ and $\mu$ 
is the product measure $\nu\times \eta$.  Then with 
$\mathfrak{A} = L^{\infty}(\mathcal{M},\mu)$, 
$\mathfrak{B} = L^{\infty}(\mathcal{N},\nu)$ and 
$\mathfrak{R} = L^{\infty}(\mathcal{R},\eta)$, we have 
that $\mathfrak{A} = \mathfrak{B}\otimes \mathfrak{R}$.
Let $\gamma:\mathcal{M}\to\mathcal{N}$ be the map $\gamma(x,y) = x$ 
for $x\in\mathcal{N}$ and $y\in\mathcal{R}$.  Then the composition 
(i.e. Koopman) operator $\KO_{\gamma}:\mathfrak{B}\to\mathfrak{A}$ 
is given by
\begin{equation}
(\KO_{\gamma}g)(x,y) = g(\gamma(x,y)) = g(x) = (g\otimes\identity)(x,y)
\nonumber
\end{equation}
for any $g\in\mathfrak{B}$.  Let $\pi:\mathfrak{A}\to\mathfrak{B}$ be given by
\begin{align}
(\pi f) & = (\id_{\mathfrak{B}}\otimes \rho_{\mathfrak{R}})(f)\nonumber\\
(\pi f)(x) & =\int_{\mathcal{R}}p(y)f(x,y)\,\rmd \eta(y)\nonumber
\end{align}
where $\rho_{\mathfrak{R}}\in\mathcal{S}(\mathfrak{R})$ is a 
normal state and $p$ is the corresponding probability density 
function on $(\mathcal{R}, \eta)$.  Then $\Prj = \KO_{\gamma}\pi$ is 
a conditional expectation on $\mathfrak{A}$ with image isomorphic 
to $\mathfrak{B}$, and $\pi\KO_{\gamma}$ is the identity morphism 
on $\mathfrak{B}$.

It is now straightforward to identity the predual operators 
(with the predual of $L^{\infty}(\mathcal{M},\mu)$ identified 
with $L^{1}(\mathcal{M},\mu)$, and likewise for 
$(\mathcal{N},\nu)$ and $(\mathcal{R},\eta)$)
\begin{align}
(\PD{\KO_{\gamma}}\phi)(x) & = \int_{\mathcal{R}}\phi(x,y)\,\rmd\eta(y),\nonumber\\
(\PD{\pi}\psi)(x,y) &= (\psi\otimes p)(x,y) = \psi(x)p(y),\nonumber\\
(\PD\Prj\phi)(x,y) &= p(y)\int_{\mathcal{R}}\phi(x',y)\,\rmd\eta(x'),\nonumber
\end{align}
for any $\phi\in\PD{\mathfrak{A}}$, $\psi\in\PD{\mathfrak{B}}$.  Then the NMZ equation \eqref{eqn:PushForwardStateNMZReduced} becomes
\begin{align}
& \frac{\rmd}{\rmd t}\sigma_{t}(x)  = \int_{\mathcal{R}}\PD{\LV[t]}(\sigma_{t}\otimes p)(x,y)\,\rmd \eta(y) \notag\\
 & + \int_{\mathcal{R}}\left[\PD{\LV[t]}\int_{0}^{t}\TOL e^{\int_{s}^{t}\PD{\PrjC}\PD{\LV[\tau]}\,\rmd \tau}\PD{\PrjC}\PD{\LV[s]}(\sigma_{s}\otimes p)\,\rmd s\right](x,y)\,\rmd \eta(y).\label{eqn:classicalSubsystemReduction}
\end{align}
\end{example}

\begin{remark}
The projection defined in Example 4 may be thought of 
as a generalization of Chorin's conditional expectation \cite{Chorin2000} 
when $\mathcal{N}\simeq \mathbb{R}^{n}$ and $\mathcal{R}\simeq\mathbb{R}^{r}$, 
with $\nu$ and $\eta$ the Lebesgue measures on these spaces.  
It is also a commutative (i.e., classical) example of the problem 
of reducing dynamics to a subsystem of interest that arises in the theory 
of open quantum systems and quantum information theory.  
\end{remark}

The steps necessary to undertake dimension reduction using the NMZ 
formalism are outlined in Algorithm \ref{alg:NMZ}.  
The last step obviously hides many important 
details involving approximation of memory integrals, noise terms  
and implementation. These details are beyond the scope of the 
present paper and we refer to 
\cite{Stinis2015,Chorin2000,Stinis2007,Venturi2014,Venturi2016} 
(see also \cite{Zhu2016}).  
In general, solving the NMZ equations is a very challenging task 
that implicitly requires propagation of all information of the system.  Only by 
suitable (typically problem-class-dependent) approximations and efficient 
numerical algorithms, can these equations be rendered tractable.  Except under 
strong assumptions (e.g. scale-separation), these issues persist and 
present serious challenges to the development of efficient and accurate 
solution methods.
\begin{algorithm}[H]
\caption{NMZ algorithm}\label{alg:NMZ}
\begin{algorithmic}[1]
% \Procedure{MyProcedure}{}
\State Identify the system $\dot{x} = F(x,\xi,t).$
\State Decide on the algebra of observables, 
typically either $\mathfrak{A} = L^\infty(\mathcal{M},\mu)$ or 
$\mathfrak{A} = C_{0}(\mathcal{M})$.
\State Identify the subalgebra of interest 
$\mathfrak{B} \subset\mathfrak{A}$.
\State Choose a faithful tracial state $\rho$ of $\mathfrak{A}$.  This is typically equivalent to choosing a strictly positive probability measure $\hat{\rho}$ on $\mathcal{M}$.
\State Extract a conditional expectation $\Prj$ as the unique 
projection $\Prj:\mathfrak{A}\to\mathfrak{B}\subset\mathfrak{A}$ 
satisfying $\rho(AB) = \rho(\Prj(A)B)$ for all $A\in\mathfrak{A}$ 
and $B\in\mathfrak{B}$.  
In the typical case, this is 
$$\int_{\mathcal{M}}A(x)B(x)\,d\hat{\rho}(x)  = 
\int_{\mathcal{M}}\Prj(A)(x)B(x)\,d\hat{\rho}(x).$$
\State Set up the NMZ equations \eqref{eqn:NMZforObservables} and/or \eqref{eqn:NMZforStates}.
\State Solve, by approximating the memory integral and/or noise term as needed.
\end{algorithmic}
\end{algorithm}
In the next section we study an example of the NMZ formalism 
applied to quantum systems, in which the observable algebra 
is non-commutative.

\subsection{Reduced-Order Quantum Dynamics}
\label{sec:QuantumExample}
Let $\HA$ and $\HB$ be Hilbert spaces for two quantum systems.  
Then the total Hilbert space for the two systems together 
is $\HAB$.  We will look at the case where we know the 
Hamiltonian dynamics of the composite system, but want to 
understand the dynamics of system $A$ alone.  This is a 
common starting point for the study of open quantum systems,
where system $B$ represents the (typically nuisance) environment 
from which we cannot entirely decouple our system $A$ of interest.  
Let $\mathfrak{A} = \mathcal{B}(\HAB)$ and 
$\mathfrak{B} = \mathcal{B}(\HA)$, 
where $\mathcal{B}(\HH)$ denotes the (noncommutative) von Neumann 
algebra of all bounded linear operators on $\HH$ with the operator
norm $\|X\|_{\infty} = \sup_{\|v\|=1}\|Xv\|$.  Note that that 
predual of $\mathcal{B}(\HH)$ is $\mathcal{T}(\HH)$, the Banach 
space of trace class bounded linear operators on $\HH$ with 
norm $\|A\|_{1} = \Tr|A|$, where $|A| = \sqrt{A^{\dag} A}$.  
As suggested by the definitions (and subscripts), $\mathcal{B}(\HH)$ 
and $\mathcal{T}(\HH)$ are in some ways analogous to $L^{\infty}$ 
and $L^{1}$ function spaces.

We consider the isometric $*$-homomorphism $\KO:\mathfrak{B}\to\mathfrak{A}$ 
given by $\KO(X) = X\otimes\identity$.  The predual of $\KO$ is 
then $\PD{\KO}(Y) = \TrB(Y)$, the partial trace of $Y$ over $\HB$. 
We seek a pseudoinverse of $\KO$, i.e. a linear
map $\pi:\mathfrak{A}\to\mathfrak{B}$ such that $\KO\pi\KO = 
\KO$ and $\pi \KO \pi = \pi$.  It is easily verified that, for 
any fixed $\rhoB\in\mathcal{S}(\mathcal{B}(\HB))$,  
$\pi(X) = \TrB(X\identity\otimes\rhoB)$ is such a pseudoinverse.  
Indeed, with this choice, $\pi\KO = \id_{\mathfrak{B}}$.  
This choice brings us in line with [\onlinecite[\S 9.1]{Breuer2002}].  
Then $\PD{\pi}:\PD{\mathfrak{A}}\to\PD{\mathfrak{B}}$ is 
$\PD{\pi}(A) = A\otimes\rhoB$.  Let $\Prj$ be the conditional
expectation $\KO \pi$ on $\mathfrak{A}$, so that 
$\Prj(X) = \TrB(X\identity\otimes\rhoB)\otimes\identity$ projects 
onto the subalgebra $\mathfrak{B}\otimes\identity\subset\mathfrak{A}$, 
and $\PD{\Prj}(A) = \TrB(A)\otimes\rhoB$ projects onto 
$\PD{\mathfrak{B}}\otimes\rhoB\subset\PD{\mathfrak{A}}$.
Then, letting $\sigma(t) = \PD{\KO}\rho(t)$, \eqref{eqn:PushForwardStateNMZ} yields:
\begin{align}
& \frac{\rmd}{\rmd t}\sigma(t) = \TrB\big[\PD{\LF[t]}\sigma(t)\otimes\rhoB\big] \notag\\
& \quad + \TrB\left[\PD{\LF[t]} \TOL e^{\int_{0}^{t}\PD{\PrjC}\PD{\LF[\tau]}\,\rmd\tau}\PD{\PrjC}\rho_{0}\right] \notag\\
& \quad + \TrB\left[\PD{\LF[t]}\int_{0}^{t}\TOL e^{\int_{s}^{t}\PD{\PrjC}\PD{\LF[\tau]}\,\rmd\tau}\PD{\PrjC}\PD{\LF[s]}\sigma(s)\otimes\rhoB\,\rmd s\right].
\end{align}
This generalized Langevin equation should be compared with the 
classical subsystem reduction in \eqref{eqn:classicalSubsystemReduction}.  In many cases, further assumptions and approximations are made in order to eliminate the memory, yielding a Markovian master equation \cite{Gorini1976}.

\section{Application to Integrable Systems}
\label{sec:AnalyticExample}
In this section we study the dynamics of simple integrable systems
on $\mathrm{SU}(2)$ and $\mathrm{SO}(3)$ and work through an analytic example in both the observable (Heisenberg) and state (Schr\"odinger) 
pictures.  

\subsection{Integrable System on $\mathrm{SU}(2)$: The Heisenberg Picture}

Let $\mathcal{M} = \mathrm{SU}(2)$ with normalized Haar 
measure $\mu$, and consider the observable 
$g(U) = \Tr(U)/2$ in the Banach algebra 
$\mathfrak{A} = L^{\infty}(\mathrm{SU}(2),\mu)$.  
Then $g = g^{\ast}$ ($g$ is real-valued) and $\|g\| = 1$.  
Consider the dynamical system
\begin{equation}
	\dot{U} = -iHU,
	\label{eqn:ExampleSU2ODE}
\end{equation}
where $H\in\mathbb{C}^{2\times2}$ is Hermitian 
and $\Tr(H) = 0$, so that, by the Cayley-Hamilton 
theorem, $H^{2} = -\identity \det{H}$.  
While \eqref{eqn:ExampleSU2ODE} 
resembles a quantum mechanical model of a 2-state 
system in some respects, we will not think of it in those terms, 
but only as a simple ODE defined on a 3-dimensional compact 
manifold. In particular, the observable algebra we consider hereafter 
is a commutative algebra of $\mathbb{C}$-valued functions 
on $\mathrm{SU}(2)$, rather than the non-commutative algebra $\mathcal{B}(\mathbb{C}^{2})$. 
Then for any differentiable $f\in L^{\infty}(\mathrm{SU}(2),\mu)$,
\begin{equation}
(\LV f)(U) = \rmd f(-iHU).
\end{equation}
Recall that the class functions on $\mathrm{SU}(2)$ are 
those functions $f:\mathrm{SU}(2)\to\mathbb{C}$ such 
that $f(U) = f(\Omega U \Omega^{\dag})$ for all 
$\Omega\in\mathrm{SU}(2)$, i.e. the functions that are 
constant on conjugacy classes.  Let $\mathfrak{B}\subset\mathfrak{A}$ 
be the von Neumann subalgebra of class functions 
within $L^{\infty}(\mathrm{SU}(2),\mu)$.  This subalgebra 
captures the observables on $\mathrm{SU}(2)$ that depend only 
on the spectra of the unitary operators.  And, since in 
$\mathrm{SU}(2)$ the spectrum of $U$ is determined by $\Tr(U)$, 
the class functions are exactly the functions that factor over $g$, 
so that $\mathfrak{B}$ is the von Neumann subalgebra 
of $\mathfrak{A}$ generated by $g$.
We take as the projection operator onto $\mathfrak{B}$ the 
conditional expectation
\begin{equation}
(\Prj f)(U) = \int_{\mathrm{SU}(2)}f(\Omega U\Omega^{\dag})\,\rmd\mu(\Omega).
\label{eq:proj}
\end{equation}
The NMZ equation for $g$ is  
\begin{align}
	\frac{d}{dt}g(t) & = e^{t\LV}\Prj\LV g + R(t) + 
	\int_{0}^{t} e^{(t-s)\LV} \Prj\LV R(s)\,\rmd s
\end{align}
where
\begin{equation}
R(t) = e^{t(\identity-\Prj)\LV}(\identity - \Prj)\LV g.
\end{equation}
To begin making this MZ equation more concrete, 
observe first that
\begin{align}
(\LV g)(U) &= -\frac{i}{2}\Tr(HU)\nonumber\\
(\LV^{2}g)(U) &= -\frac{1}{2}\Tr(H^{2}U) = \det(H)g(U),\nonumber
\end{align}
so that $g$ and $\LV g$ span an invariant subspace of $\LV$.
Moreover, $g$ is a class function of $\mathrm{SU}(2)$, so $\Prj g = g$, 
and 
\begin{align}
(\Prj \LV g)(U) & = \int_{\mathrm{SU}(2)}
(\LV g)(\Omega U\Omega^{\dag})\,\rmd\mu(\Omega)\nonumber\\
& = -\frac{i}{2}\int_{\mathrm{SU}(2)}\Tr(H\Omega U\Omega^{\dag})\,
\rmd\mu(\Omega)\nonumber\\
& = -\frac{i}{2}\Tr\left[\left(\int_{\mathrm{SU}(2)}\Omega^{\dag}
H\Omega \,\rmd\mu(\Omega)\right)U\right]\nonumber\\
& = -i\Tr(H)g(U) = 0
\end{align}
for all $U\in\mathrm{SU}(2)$ since $\Tr(H) = 0$.
Then $\Span\{g,\LV g\}$ is also an invariant subspace of $\Prj$ 
(and therefore also of $\identity - \Prj$).
It follows that
\begin{equation}
\big[(\identity-\Prj)\LV \big]^{k}g = 
\begin{cases} 
g & k = 0\\\LV g  
&  k=1\\0 & k\geq 2
\end{cases},
\end{equation}
and therefore
\begin{align}
R(t) & = e^{t(\identity - \Prj)\LV}(\identity - \Prj)\LV g \nonumber\\
&= \sum_{k=0}^{\infty}\frac{t^{k}}{k!}\big[(\identity-\Prj)\LV\big]^{k+1}g \notag\\
& = \LV g.
\end{align}
This means that $R(t)(U_{0}) = R(U_{0}) = -i\Tr(HU_{0})$ 
is constant in time and that $\Prj \LV R = \det(H) g$.
The NMZ equation then reduces to
\begin{equation}
\frac{d}{dt}g(t) = R + \det(H)\int_{0}^{t}g(t-s)\,\rmd s.
\label{eqn:ExampleMZResult}
\end{equation}
Now, since the ODE \eqref{eqn:ExampleSU2ODE} is linear, 
we can solve it exactly.  This is particularly simple since, 
as we observed above, $H^{2} = -\det(H)\identity$.  
We therefore immediately get
\begin{align}
U(t) & = e^{-itH}U_{0} = \cos(\lambda t)U_{0} - i\frac{\sin(\lambda t)}{\lambda}HU_{0},
%& = \begin{bmatrix}\cosh(\sqrt{2}t)x_{1}(0) + \frac{\sinh{\sqrt{2}t}}{\sqrt{2}}\big(x_{1}(0) + x_{2}(0)\big)\\
%\cosh(\sqrt{2}t)x_{2}(0) + \frac{\sinh{\sqrt{2}t}}{\sqrt{2}}\big(x_{1}(0) - x_{2}(0)\big)
%\end{bmatrix}.
\end{align}
where $\lambda = \sqrt{-\det{H}}$, so that $\sigma(H) = \{\pm\lambda\}$.
Thus,
\begin{align}
	g(t) &= \cos(\lambda t)g+ \frac{\sin(\lambda t)}{\lambda}\LV g.
\label{eq:sol}
\end{align}
It is easy to show by direct substitution that \eqref{eq:sol} 
is indeed the solution to the integro-differential 
equation \eqref{eqn:ExampleMZResult}, as desired.

\subsection{Integrable System on $\mathrm{SU}(2)$: The Schr\"odinger Picture}
We now turn to the problem of solving the 
predual NMZ equation for the evolution of a reduced normal state. 
Note that in the present example, this is a classical reduced-order 
probability distribution function on $\mathrm{SU}(2)$, not a density 
matrix as would be typical in a quantum mechanical setting.
We  consider the subalgebra $\mathfrak{B}\subset\mathfrak{A}$ 
of bounded class functions, i.e. $f\in\mathfrak{A}$ such 
that $f(\Omega U\Omega^{\dag}) = f(U)$ for 
all $U, \Omega\in \mathrm{SU}(2)$.  The projector \eqref{eq:proj}
$\Prj\in\mathcal{B}(\mathfrak{A})$ given in \eqref{eq:proj} 
has predual $\PD{\Prj}\in\mathcal{B}(\PD{\mathfrak{A}})$ with 
the same form, i.e., 
\begin{equation}
(\PD{\Prj}\rho)(U) = \int_{\mathrm{SU}(2)}\rho(\Omega^{\dag}U\Omega)\,\rmd\mu(\Omega).
\end{equation} 
Here, and throughout this section, we'll freely use the isomorphism 
$\PD{L^{\infty}(\mathcal{M},\mu)}\simeq L^{1}(\mathcal{M},\mu)$ to identify 
functionals in $\PD{\mathfrak{A}}$ with $\mu$-integrable functions.  Likewise, 
the predual Liouvillian $\PD{\LV}$ takes almost the same form as $\LV$, namely
\begin{equation}
(\PD{\LV}\rho)(U) = \rmd_{U}\rho(iHU).
\end{equation}
Now, suppose we take as initial state the PDF $\rho_{0}(U) = \Tr(U)^{2}$.  This is positive valued on $\mathrm{SU}(2)$ because $\Tr$ is real-valued on $\mathrm{SU}(2)$, and it is normalized because
\begin{align}
\int_{\mathrm{SU}(2)}\Tr(U)^{2}\,\rmd\mu(U) & = \Tr\left[\int_{\mathrm{SU}(2)}U\otimes U\,\rmd\mu(U)\right] \notag\\
& = \Tr[\frac{1}{2}(\identity - \textsc{swap})] = 1,
\end{align}
where $\textsc{swap}$ is the operator on $\HH\otimes\HH$ that 
permutes the two subsystems, i.e., 
$\textsc{swap}\big(|\psi\rangle\otimes|\phi\rangle\big) = 
|\phi\rangle\otimes\psi\rangle$.
Because $\PD{\Prj}\rho_{0} = \rho_{0}$, the NMZ equation \eqref{eqn:NMZforStates} 
that we wish to solve reduces to
\begin{align}
\label{eqn:UnitaryStateExampleNMZ}
	\frac{\rmd}{\rmd t}\PD{\Prj}\rho(t) &= \PD{\Prj}\PD{\LV}\PD{\Prj}\rho(t) \notag\\
	& +  \PD{\Prj}\PD{\LV}\int_{0}^{t} e^{(t-s)\PD{\PrjC}\PD{\LV}}\PD{\PrjC}\PD{\LV}\PD{\Prj}\rho(s)\,\rmd s.
\end{align}
Next, we look for suitable matrix representations
of $\PD{\Prj}$ and $\PD{\LV}$. To this end, 
consider the linearly independent family of 
functions 
\begin{equation}
\{1,\rho_{0} = 4g^{2},4g\PD{\LV}g, 4(\PD{\LV}g)^{2}\},
\label{family}
\end{equation}
where the observable $g(U) = \Tr(U)/2$ is as in the previous section.  
It easy to show that the space spanned by these functions is 
invariant for $\PD{\LV}$ and $\PD{\Prj}$. 
It can also be verified that $\PD{\LV}$ and $\PD{\Prj}$ have the following matrix 
representations relative to \eqref{family}
\begin{align}
\PD{\LV} & \simeq \begin{bmatrix}0 & 0 & 0 & 0\\
0 & 0 & -\lambda^{2} & 0 \\ 0 & 2 & 0 & -2\lambda^{2}\\ 
0 & 0 & 1 & 0\end{bmatrix}\\
\PD{\Prj} & \simeq \begin{bmatrix}1 & 0 & 0 & \frac{4}{3}\lambda^{2}\\
0 & 1 & 0 & -\frac{1}{3}\lambda^{2} \\ 
0 & 0 & 0 & 0\\ 0 & 0 & 0 & 0\end{bmatrix}. %\\
\end{align}
Therefore,  
\begin{widetext}
\begin{align}
\PD{\Prj}\PD{\LV} & \simeq \begin{bmatrix}0 & 0 & \frac{4}{3}\lambda^{2} & 0\\ 
0 & 0 & -\frac{4}{3}\lambda^{2} & 0\\ 0 & 0 & 0 & 0\\ 0 & 0 & 0 &0\end{bmatrix},\\
\PD{\Prj}\PD{\LV}e^{(t-s)\PD{\PrjC}\PD{\LV}}\PD{\PrjC}\PD{\LV} & \simeq \begin{bmatrix}0 & \frac{8}{3}\lambda^{2}\cos\left(\frac{2\lambda(t-s)}{\sqrt{3}}\right) & \frac{8}{3\sqrt{3}}\lambda^{3}\sin\left(\frac{2\lambda(t-s)}{\sqrt{3}}\right) & -\frac{8}{3}\lambda^{4}\cos\left(\frac{2\lambda(t-s)}{\sqrt{3}}\right)\\
0 & -\frac{8}{3}\lambda^{2}\cos\left(\frac{2\lambda(t-s)}{\sqrt{3}}\right) & -\frac{8}{3\sqrt{3}}\lambda^{3}\sin\left(\frac{2\lambda(t-s)}{\sqrt{3}}\right) & \frac{8}{3}\lambda^{4}\cos\left(\frac{2\lambda(t-s)}{\sqrt{3}}\right)\\
 0 & 0 & 0 & 0\\ 0 & 0 & 0 &0\end{bmatrix}.
\end{align}
\end{widetext}
Since $\PD{\Prj}\rho(0) \simeq \begin{bmatrix}0 & 1 & 0 & 0\end{bmatrix}^{\rmT}$,
it follows that $\PD{\Prj}\rho(t) = a(t)1 + b(t)\rho_{0}$, so we can 
reduce to the 2-dimensional invariant subspace spanned by $1$ 
and $\rho_{0}$, yielding
\begin{align}
\PD{\Prj}\PD{\LV} & \simeq \begin{bmatrix}0 & 0 \\ 0 & 0\end{bmatrix},\\
\PD{\Prj}\PD{\LV}e^{(t-s)\PD{\PrjC}\PD{\LV}}\PD{\PrjC}\PD{\LV} & \simeq \begin{bmatrix}0 & \frac{8}{3}\lambda^{2}\cos\left(\frac{2\lambda(t-s)}{\sqrt{3}}\right)\\
0 & -\frac{8}{3}\lambda^{2}\cos\left(\frac{2\lambda(t-s)}{\sqrt{3}}\right)\end{bmatrix}.
\end{align}
The NMZ equation \eqref{eqn:UnitaryStateExampleNMZ} then becomes
\begin{equation}
\begin{bmatrix}a'(t)\\b'(t)\end{bmatrix} = \frac{8}{3}\lambda^{2}\int_{0}^{t}\cos\left(\frac{2\lambda(t-s)}{\sqrt{3}}\right)\begin{bmatrix}0 & 1\\0 & -1\end{bmatrix}\begin{bmatrix}a(s)\\b(s)\end{bmatrix}\,\rmd s.
\end{equation}
Note that $a(0)+b(0) = 1$ (since $a(0) = 0$ and $b(0) = 1$) 
and $a'(t)+b'(t) = 0$, so that $a(t)+b(t) = 1$ for all $t\geq 0$. 
Moreover, $b(t)$ is described by the integro-differential equation
\begin{equation}
b'(t) = -\frac{8}{3}\lambda^{2}\int_{0}^{t}\cos\left(\frac{2\lambda(t-s)}{\sqrt{3}}\right)b(s)\,\rmd s.
\end{equation}
Differentiating this expression twice more, we find that
\begin{align}
b''(t) & = -\frac{8}{3}b(t) -\frac{16}{3\sqrt{3}}\lambda^{3}
\int_{0}^{t}\sin\left(\frac{2\lambda(t-s)}{\sqrt{3}}\right)b(s)\,\rmd s\nonumber\\
b'''(t) & = -\frac{8}{3}b'(t) -\frac{32}{9}\lambda^{4}\int_{0}^{t}
\cos\left(\frac{2\lambda(t-s)}{\sqrt{3}}\right)b(s)\,\rmd s\nonumber\\
& = -4\lambda^{2}b'(t),
\end{align}
so that
%\begin{subequations}
\begin{align}
%b'(t) & = \alpha\sin(2\lambda t) + \beta\cos(2\lambda t)\\
b(t) & = C + \gamma\sin(2\lambda t) + \kappa\cos(2\lambda t).
\end{align}
%\end{subequations}
Using the initial conditions $b(0) = 1$, $b'(0) = 0$, $b''(0) = 
-8\lambda^{2}/3$ (the last two are clear from the integro-differential 
equations for $b'$ and $b''$ above), we conclude that
\begin{equation}
b(t) = \frac{1}{3}(2\cos(2\lambda t) + 1)
\end{equation}
so that
\begin{align}
\PD{\Prj}\rho(t) & = \frac{2}{3}(1-\cos(2\lambda t))1 + 
\frac{1}{3}(2\cos(2\lambda t) + 1)\rho_{0} \nonumber\\
& = \rho_{0} + \frac{2}{3}(1-\cos(2\lambda t))(1-\rho_{0}).
\end{align}
This solution can be also obtained by exponentiating $\PD{\LV}$ in the 
4-dimensional invariant subspace, yielding
\begin{align}
\PD{\Prj}e^{t\PD{\LV}} & \simeq \begin{bmatrix}1 
& \frac{4}{3}\sin^{2}(\lambda t) & 
\frac{2}{3}\lambda \sin(2\lambda t) & \frac{4}{3}\lambda^{2}\cos^{2}(\lambda t)\\ 
0 & \frac{1+2\cos(2\lambda t)}{3} & -\frac{2}{3}\lambda\sin(2\lambda t) & 
\lambda^{2}\frac{1-2\cos(2\lambda t)}{3}\\0 & 0 & 0 & 0\\ 0 & 0 & 0 &0
\end{bmatrix}.\nonumber
\end{align}

\subsection{Integrable System on $\mathrm{SO}(3)$: The Heisenberg Picture}
Let $\mathcal{M} = \mathrm{SO}(3)$ with normalized Haar 
measure $\mu$, and consider the observable 
$g(\mathcal{O}) = \Tr(\mathcal{O})/3$ in the Banach algebra 
$\mathfrak{A} = L^{\infty}(\mathrm{SO}(3),\mu)$.  
Then $g = g^{\ast}$ ($g$ is real-valued) and $\|g\| = 1$.  
Consider the dynamical system
\begin{align}
	\frac{d}{dt}{\mathcal{O}}(t) & = X\mathcal{O}(t), % & \mathcal{O}(0) & = \identity ,
	\label{eqn:ExampleSO3ODE}
\end{align}
where $X\in\mathrm{SO}(3)$ is skew-symmetric of the form
\begin{align}
	X = \begin{bmatrix}0 & -z & y\\ z & 0 & -x\\-y & x & 0\end{bmatrix}
\end{align}
with $\Tr(X) = 0$ and $\det(X) = 0$, so that, by the Cayley-Hamilton 
theorem, $X^{3} = -r^{2}X$, where $r^{2} = x^{2}+y^{2}+z^{2}$.  The 
dynamical system \eqref{eqn:ExampleSO3ODE} is then a 
simple ODE on a 3-dimensional compact manifold, which describes a one-
parameter semigroup of orthogonal operators $\mathcal{O}(t)$ effecting a 
constant-rate rigid rotation about the axis along $(x,y,z)$.  Then for any 
differentiable $f\in L^{\infty}( \mathrm{SO}(3),\mu)$,
\begin{equation}
(\LV f)(\mathcal{O}) = \rmd f(X\mathcal{O}).
\end{equation}
Recall that the class functions on $\mathrm{SO}(3)$ are 
those functions $f:\mathrm{SO}(3)\to\mathbb{C}$ such 
that $f(\mathcal{O}) = f(\Omega \mathcal{O} \Omega^{\rmT})$ for all 
$\Omega\in\mathrm{SO}(3)$, i.e. the functions that are 
constant on conjugacy classes.  Let $\mathfrak{B}\subset\mathfrak{A}$ 
be the von Neumann subalgebra of class functions 
within $L^{\infty}(\mathrm{SO}(3),\mu)$.  This subalgebra 
captures the observables on $\mathrm{SO}(3)$ that depend only 
on the spectra of the orthogonal operators.  Moreover, since in 
$\mathrm{SO}(3)$ the spectrum of $\mathcal{O}$ is determined by $\Tr(\mathcal{O})$, the class functions are exactly the 
functions that factor over $g$, 
so that $\mathfrak{B}$ is the von Neumann subalgebra 
of $\mathfrak{A}$ generated by $g$.
We take as the projection operator onto $\mathfrak{B}$ the 
conditional expectation
\begin{equation}
(\Prj f)(\mathcal{O}) = \int_{\mathrm{SO}(3)}f(\Omega \mathcal{O}\Omega^{\rmT})\,\rmd\mu(\Omega).
\label{eq:proj}
\end{equation}
The NMZ equation for $g$ is  
\begin{align}
	\frac{d}{dt}g(t) & = e^{t\LV}\Prj\LV g + R(t) + 
	\int_{0}^{t} e^{(t-s)\LV} \Prj\LV R(s)\,\rmd s
\end{align}
where
\begin{equation}
R(t) = e^{t(\identity-\Prj)\LV}(\identity - \Prj)\LV g.
\end{equation}
To make this NMZ equation more concrete, 
we first observe that
\begin{align}
(\LV g)(\mathcal{O}) &= \Tr(X\mathcal{O})/3,\notag\\
(\LV^{2}g)(\mathcal{O}) &= \Tr(X^{2}\mathcal{O})/3,\\
(\LV^{3}g)(\mathcal{O}) &= \Tr(X^{3}\mathcal{O})/3 = -r^{2}\Tr(X\mathcal{O})/3 = -r^{2}(\LV g)(\mathcal{O}),\notag
\end{align}
so that $g$, $\LV g$, and $\LV^{2}g$ span an invariant subspace of $\LV$ and
\begin{align}
	\LV^{k}g & = \begin{cases}g & k = 0\\(-r^{2})^{\frac{k-1}{2}}\LV g & k \text{ odd}, k\geq 1\\(-r^{2})^{\frac{k-2}{2}}\LV^{2} g & k \text{ even}, k\geq 2.\end{cases}
\end{align}
Moreover, $g$ is a class function of $\mathrm{SO}(3)$, and therefore 
we have $\Prj g = g$ and 
\begin{align}
(\Prj \LV^{k} g)(\mathcal{O}) & = \int_{\mathrm{SO}(3)}
(\LV^{k} g)(\Omega \mathcal{O}\Omega^{\rmT})\,\rmd\mu(\Omega)\nonumber\\
& = \frac{1}{3}\int_{\mathrm{SO}(3)}\Tr(X^{k}\Omega \mathcal{O}\Omega^{\rmT})\,
\rmd\mu(\Omega)\nonumber\\
& = \frac{1}{3}\Tr\left[\left(\int_{\mathrm{SO}(3)}\Omega^{\rmT}
X^{k}\Omega \,\rmd\mu(\Omega)\right)\mathcal{O}\right]\nonumber\\
& = \frac{\Tr(X^{k})}{3}g(\mathcal{O})
\end{align}
for all $\mathcal{O}\in\mathrm{SO}(3)$ and $k\geq 1$.  Since $\Tr(X) = 0$ and $\Tr(X^{2}) = -2r^{2}$, 
\begin{align}
	\Prj\LV^{k}g & = \begin{cases}g & k = 0\\0 & k \text{ odd}, k\geq 1\\\frac{2}{3}(-r^{2})^{\frac{k}{2}} g & k \text{ even}, k\geq 2.\end{cases}
\end{align}
Thus, $\Span\{g,\LV g, \LV^{2}g\}$ is also an invariant subspace of $\Prj$ 
(and therefore also of $\identity - \Prj$). This implies that 
\begin{align}
	[(\identity - \Prj)\LV]^{k} g & = \begin{cases}g & k = 0\\(-r^{2}/3)^{\frac{k-1}{2}}\LV g & k \text{ odd}, k\geq 1 \\(-r^{2}/3)^{\frac{k-2}{2}}(\LV^{2} g + (2/3)r^{2}g) & k \text{ even}, k \geq 2.\end{cases}
\end{align}
and therefore
\begin{align}
	R(t) &= e^{t(\identity - \Prj)\LV}(\identity - \Prj)\LV g\nonumber\\
%	& = \sum_{k=0}^{\infty}\frac{t^{k}}{k!}[(\identity - \Prj)\LV]^{k+1}g\nonumber\\
	& = \sum_{k=0}^{\infty}\frac{t^{2k}}{(2k)!}[(\identity - \Prj)\LV]^{2k+1}g + \sum_{k=0}^{\infty}\frac{t^{2k+1}}{(2k+1)!}[(\identity - \Prj)\LV]^{2k+2}g\nonumber\\
%	& = \sum_{k=0}^{\infty}\frac{t^{2k}}{(2k)!}(-r^{2}/3)^{k}\LV g + \sum_{k=0}^{\infty}\frac{t^{2k+1}}{(2k+1)!}(-r^{2}/3)^{k}(\LV^{2} g + (2/3)r^{2}g)\nonumber\\
	& = \sum_{k=0}^{\infty}(-1)^{k}\frac{(rt/\sqrt{3})^{2k}}{(2k)!}\LV g + \frac{\sqrt{3}}{r}\sum_{k=0}^{\infty}(-1)^{k}\frac{(rt/\sqrt{3})^{2k+1}}{(2k+1)!}\left(\LV^{2} g + \frac{2}{3}r^{2}g\right)\nonumber\\
	& = \cos\left(\frac{rt}{\sqrt{3}}\right)\LV g + \frac{\sqrt{3}}{r}\sin\left(\frac{rt}{\sqrt{3}}\right)\left(\LV^{2} g + \frac{2}{3}r^{2}g\right).
\end{align}
By applying $\Prj \LV$ to $R(s)$, we obtain
\begin{align}
	\Prj \LV R(s) & = -\frac{2}{3}r^{2}\cos\left(\frac{rs}{\sqrt{3}}\right)g.
\end{align}
Thus, the NMZ equation then reduces to
\begin{equation}
\frac{d}{dt}g(t) = R(t) -\frac{2}{3}r^{2}\int_{0}^{t}\cos\left(\frac{rs}{\sqrt{3}}\right)g(t-s)\,\rmd s,
\label{eqn:ExampleMZResult}
\end{equation}
which may be solved (e.g., via Laplace transforms) to obtain
\begin{align}
g_{t} =  g + \frac{\sin(rt)}{r}\LV g +\frac{1-\cos(rt)}{r^{2}}\LV^{2} g.
\label{eq:SO3sol}
\end{align}
Of course, since the ODE \eqref{eqn:ExampleSO3ODE} is linear, 
we can solve it exactly.  This is particularly simple since, 
as we observed above, $X^{3} = -r^{2}X$.  
We therefore find
\begin{align}
	\mathcal{O}(t) & = e^{tX}\mathcal{O}_{0}\nonumber\\
%	& = \sum_{k=0}^{\infty}\frac{t^{k}}{k!}X^{k}\mathcal{O}_{0}\nonumber\\
	& = \sum_{k=0}^{\infty}\frac{t^{2k}}{(2k)!}X^{2k}\mathcal{O}_{0} + \sum_{k=0}^{\infty}\frac{t^{2k+1}}{(2k+1)!}X^{2k+1}\mathcal{O}_{0}\nonumber\\
%	& = \mathcal{O}_{0} + \sum_{k=0}^{\infty}\frac{t^{2k+1}}{(2k+1)!}(-r^{2})^{k}X\mathcal{O}_{0} + \sum_{k=1}^{\infty}\frac{t^{2k}}{(2k)!}(-r^{2})^{k-1}X^{2}\mathcal{O}_{0}\nonumber \\
	& = \mathcal{O}_{0} + \frac{1}{r} \sum_{k=0}^{\infty}(-1)^{k}\frac{(rt)^{2k+1}}{(2k+1)!}X\mathcal{O}_{0} - \frac{1}{r^{2}} \sum_{k=1}^{\infty}(-1)^{k}\frac{(rt)^{2k}}{(2k)!}X^{2}\mathcal{O}_{0}\nonumber\\
	& = \mathcal{O}_{0} + \frac{\sin(rt)}{r} X\mathcal{O}_{0} + \frac{1- \cos(rt)}{r^{2}} X^{2}\mathcal{O}_{0}.
\end{align}
Thus,
\begin{align}
	g(t)(\mathcal{O}_{0}) &= \Tr(\mathcal{O}(t))/3\nonumber\\
	& = \Tr(\mathcal{O}_{0})/3 + \frac{\sin(rt)}{r} \Tr(X\mathcal{O}_{0})/3 + \frac{1- \cos(rt)}{r^{2}} \Tr(X^{2}\mathcal{O}_{0})/3\nonumber\\
	& = g(\mathcal{O}_{0}) + \frac{\sin(rt)}{r} (\LV g)(\mathcal{O}_{0}) + \frac{1- \cos(rt)}{r^{2}} (\LV^{2} g)(\mathcal{O}_{0}),
\label{eq:sol}
\end{align}
confirming the solution obtained through the NMZ formalism (compare \eqref{eq:sol} and \eqref{eq:SO3sol}).

\subsection{Integrable System on $\mathrm{SO}(3)$:  The Schr\"odinger Picture}
We now turn to the problem of solving the 
predual NMZ equation for the evolution of a reduced normal state. 
We consider the subalgebra $\mathfrak{B}\subset\mathfrak{A}$ 
of bounded class functions, i.e. $f\in\mathfrak{A}$ such 
that $f(\Omega \mathcal{O}\Omega^{\rmT}) = f(\mathcal{O})$ for 
all $\mathcal{O}, \Omega\in \mathrm{SO}(3)$.  The projector \eqref{eq:proj}
$\Prj\in\mathcal{B}(\mathfrak{A})$ given in \eqref{eq:proj} 
has predual $\PD{\Prj}\in\mathcal{B}(\PD{\mathfrak{A}})$ with 
the same form, i.e., 
\begin{equation}
(\PD{\Prj}\rho)(\mathcal{O}) = \int_{\mathrm{SO}(3)}\rho(\Omega^{\rmT}\mathcal{O}\Omega)\,\rmd\mu(\Omega).
\end{equation} 
Here, and throughout this section, we'll freely use the isomorphism $\PD{L^{\infty}(\mathcal{M},\mu)}\simeq L^{1}(\mathcal{M},\mu)$ to identify functionals in $\PD{\mathfrak{A}}$ with $\mu$-integrable functions.  Likewise, the predual Liouvillian $\PD{\LV}$ takes almost the same form as $\LV$, namely
\begin{equation}
(\PD{\LV}\rho)(\mathcal{O}) = \rmd_{\mathcal{O}}\rho(X\mathcal{O}).
\end{equation}
Now, suppose we take as initial state the PDF $\rho_{0}(\mathcal{O}) = \Tr(\mathcal{O})^{2}$.  This is positive valued on $\mathrm{SO}(3)$ because 
the trace operator is real-valued on $\mathrm{SO}(3)$, and it is normalized
\begin{align}
\int_{\mathrm{SO}(3)}\Tr(\mathcal{O})^{2}\,\rmd\mu(\mathcal{O}) & = \Tr\left[\int_{\mathrm{SO}(3)}\mathcal{O}\otimes \mathcal{O}\,\rmd\mu(\mathcal{O})\right] = 1.
\end{align}
This follows from the fact that 
\begin{align}
\int_{\mathrm{SO}(3)}\mathcal{O}\otimes \mathcal{O}\,\rmd\mu(\mathcal{O})
\end{align}
is the orthogonal projection onto the subspace spanned by
\begin{align}
	\mathbf{e}_{1}\otimes\mathbf{e}_{1} + \mathbf{e}_{2}\otimes\mathbf{e}_{2} + \mathbf{e}_{3}\otimes\mathbf{e}_{3}.
\end{align}
The NMZ equation \eqref{eqn:NMZforStates} for the PDF $\rho(t)$ 
takes the form 
\begin{align}
\label{eqn:OrthogonalStateExampleNMZ}
	\frac{\rmd}{\rmd t}\PD{\Prj}\rho(t) &= \PD{\Prj}\PD{\LV}\PD{\Prj}\rho(t) +  \PD{\Prj}\PD{\LV}\int_{0}^{t} e^{(t-s)\PD{\PrjC}\PD{\LV}}\PD{\PrjC}\PD{\LV}\PD{\Prj}\rho(s)\,\rmd s.
\end{align}
Next, we look for suitable matrix representations
of $\PD{\Prj}$ and $\PD{\LV}$. 
%To this end, consider the linearly independent family of 
%functions 
%\begin{equation}
%\{1,\rho_{0} = 4g^{2},4g\PD{\LV}g, 4(\PD{\LV}g)^{2}\},
%\label{family}
%\end{equation}
%where the observable $g(\mathcal{O}) = \Tr(\mathcal{O})/3$ is as in the previous section.  
%It easy to show that the space spanned by these functions is 
%invariant under both $\PD{\LV}$ and $\PD{\Prj}$. 
%Note that
%\begin{align}
%	A\otimes B\mapsto\int_{\mathrm{SO}(3)}\Omega^{\rmT}A\Omega\otimes\Omega^{\rmT}B\Omega\,d\mu(\Omega)
%\end{align}
%is the orthogonal projection (with respect to the Hilbert-Schmidt inner product) onto the subspace spanned by $(\identity, \textsc{swap}, G_{x}\otimes G_{x} + G_{y}\otimes G_{y} + G_{z}\otimes G_{z})$.  So, this projection is
%\begin{align}
%	A\otimes B &\mapsto \frac{\Tr(A)\Tr(B)}{3}\identity + \frac{\Tr(AB)}{3}\textsc{swap} \notag\\
%	& \quad + \frac{\Tr(AG_{x})\Tr(BG_{x}) + \Tr(AG_{y})\Tr(BG_{y}) + \Tr(AG_{z})\Tr(BG_{z})}{3\sqrt{3}}\times\notag\\
%	& \qquad \times (G_{x}\otimes G_{x} + G_{y}\otimes G_{y} + G_{z}\otimes G_{z}).
%\end{align}
%Thus,
%\begin{align}
%X\otimes X & \mapsto -\frac{r^{2}}{3}(G_{x}\otimes G_{x} + G_{y}\otimes G_{y} + G_{z}\otimes G_{z})\\
%X\otimes X^{2} & \mapsto 0\\
%X^{2}\otimes X^{2} & \mapsto 	-\frac{r^{4}}{15}\big[6\identity +2\textsc{swap} + (G_{x}\otimes G_{x} + G_{y}\otimes G_{y} + G_{z}\otimes G_{z})\big].
%\end{align}
%%
% It may be verified that the 10-dimensional space spanned by
To this end, consider the 10-dimensional space spanned by 
the linearly independent functions 
\begin{align}
	\{1, g, g^{2}, \LV(g), \LV^{2}(g), g\PD{\LV}(g), g\PD{\LV}^{2}(g), \PD{\LV}(g)^{2}, \PD{\LV}(g)\PD{\LV}^{2}(g), (\PD{\LV}^{2}(g))^{2}\}.
	\label{eq:span}
\end{align}
It easy to show that the space spanned by these functions is 
invariant under $\PD{\LV}$ and $\PD{\Prj}$, 
and $\PD{\PrjC} = \identity - \PD{\Prj}$.  
With respect to basis elements \eqref{eq:span}, 
these operators may be represented as
\begin{align}
	\PD{\LV} & \simeq \begin{bmatrix}
		0 & 0 & 0 & 0 & 0 & 0 & 0 & 0 & 0 & 0\\
		0 & 0 & 0 & 0 & 0 & 0 & 0 & 0 & 0 & 0\\
		0 & 0 & 0 & 0 & 0 & 0 & 0 & 0 & 0 & 0 \\
		0 & 1 & 0 & 0 & -r^{2} & 0 & 0 & 0 & 0 & 0\\
		0 & 0 & 0 & 1 & 0 & 0 & 0 & 0 & 0 & 0\\
		0 & 0 & 2 & 0 & 0 & 0 & -r^{2} & 0 & 0 & 0\\
		0 & 0 & 0 & 0 & 0 & 1 & 0 & 0 & 0 & 0 \\
		0 & 0 & 0 & 0 & 0 & 1 & 0 & 0 & -r^{2} & 0\\
		0 & 0 & 0 & 0 & 0 & 0 & 1 & 2 & 0 & -2r^{2}\\ 
		0 & 0 & 0 & 0 & 0 & 0 & 0 & 0 & 1 & 0\\
	\end{bmatrix}\\
	\PD{\Prj} & \simeq \begin{bmatrix}
		1 & 0 & 0 & 0 & 0 & 0 & 0 & r^{2} & 0 & \frac{1}{5}r^{4}\\
		0 & 1 & 0 & 0 & -\frac{2}{3}r^{2} & 0 & 0 & 2r^{2} & 0 & -\frac{2}{5}r^{4}\\
		0 & 0 & 1 & 0 & 0 & 0 & -\frac{2}{3}r^{2} & -3r^{2} & 0 & \frac{21}{5}r^{4} \\
		0 & 0 & 0 & 0 & 0 & 0 & 0 & 0 & 0 & 0\\
		0 & 0 & 0 & 0 & 0 & 0 & 0 & 0 & 0 & 0\\
		0 & 0 & 0 & 0 & 0 & 0 & 0 & 0 & 0 & 0\\
		0 & 0 & 0 & 0 & 0 & 0 & 0 & 0 & 0 & 0 \\
		0 & 0 & 0 & 0 & 0 & 0 & 0 & 0 & 0 & 0\\
		0 & 0 & 0 & 0 & 0 & 0 & 0 & 0 & 0 & 0\\ 
		0 & 0 & 0 & 0 & 0 & 0 & 0 & 0 & 0 & 0\\
	\end{bmatrix}
\end{align}
Next, we observe that $\PD{\Prj}$, restricted to the span of \eqref{eq:span} 
has image the subspace spanned by $\{1,g,g^{2}\}$.  Using the 
fact that the spectrum (with multiplicity) 
of $\PD{\PrjC}\PD{\LV}$ on the span of \eqref{eq:span} is
\begin{align*}
\sigma(\PD{\PrjC}\PD{\LV}) = \{0,0,0,0, \pm ir/\sqrt{3}, \pm r(\alpha+i\beta), \pm r(\alpha - i\beta)\}
\end{align*}
where
\begin{align}
	\alpha & = \sqrt{\frac{7}{12} + \sqrt{\frac{67}{15}}} & \beta & = \sqrt{-\frac{7}{12} + \sqrt{\frac{67}{15}}},
\end{align}
it can be verified that, on the 3-dimensional space spanned by $\{1,g,g^{2}\}$, 
\begin{align}
\PD{\Prj}&\PD{\LV}  \simeq \begin{bmatrix}0 & 0 & 0 \\ 0 & 0 & 0\\0 & 0 & 0\end{bmatrix},
\end{align}
\begin{align}
&\PD{\Prj}\PD{\LV}e^{(t-s)\PD{\PrjC}\PD{\LV}}\PD{\PrjC}\PD{\LV}\notag\\
& \simeq \begin{bmatrix}0 & 0 & 2r^{2}\displaystyle\cosh(\alpha rt)\cos(\beta rt) + \frac{83r^{2}}{30\alpha\beta}\sinh(\alpha rt)\sin(\beta rt)\\
0 & \displaystyle -\frac{2}{3}r^{2}\cos(rt/\sqrt{3}) & \displaystyle\frac{1}{211}\left[842r^{2}\cosh(\alpha rt)\cos(\beta rt) + \frac{28319r^{2}}{30\alpha\beta}\sinh(\alpha rt)\sin(\beta rt) + 2\cos(t/\sqrt{3})\right]\\0 & 0 & 
\displaystyle -\frac{22}{3}r^{2}\cosh(\alpha rt)\cos(\beta rt) + \frac{239r^{2}}{90\alpha\beta}\sinh(\alpha rt)\sin(\beta rt)\end{bmatrix}.
\label{eq:integrand}
\end{align}
With respect to $\{1,g,g^{2}\}$, the NMZ equation \eqref{eqn:OrthogonalStateExampleNMZ} then becomes
\begin{equation}
 \frac{dA(t)}{dt} = \int_{0}^{t}
\PD{\Prj}\PD{\LV}e^{(\tau-s)\PD{\PrjC}\PD{\LV}}\PD{\PrjC}\PD{\LV}A(\tau)
d\tau,
\label{eq:intDeq} 
\end{equation}
where $A(t)=[a_1(t), a_2(t),a_3(t)]^T$ are the 
components of $\PD{\Prj}\rho(t)$ relative to $\{1,g,g^{2}\}$, 
and  $\PD{\Prj}\PD{\LV}e^{(\tau-s)\PD{\PrjC}\PD{\LV}}\PD{\PrjC}\PD{\LV}$
is the $3\times 3$ matrix given explicitly in \eqref{eq:integrand}.
The integro-differential equation \eqref{eq:intDeq} can then finally be solved 
via Langrange transforms to obtain
\begin{align}
	\PD{\Prj}\rho(t) & = \frac{18}{5}[2 - \cos(rt) - \cos(2rt)]\notag\\
	& \qquad  + \frac{4}{9495}[-19179 + 61376\cos(rt) - 42917\cos(2rt) - 32241rt\sin(rt)]g \notag\\
	& \qquad  + \frac{1}{15}[67 - 106\cos(rt) + 54\cos(2rt)]\rho_{0},
\end{align}
where $r^2=x_1^2+x_2^2+x_3^2$ and $g=\Tr(\mathcal{O})/3$.

\section{Summary}
\label{sec:Summary}
We have developed a new formulation of the 
Nakajima-Mori-Zwanzig (NMZ) method of projections 
based on operator algebras of observables and associated states. 
The new theory does not depend on the commutativity 
of the observable algebra, and therefore it is equally applicable 
to both classical and quantum systems.
%We showed that two commonly used NMZ equations have a deep connection through a duality principle between observables and states.
We established a duality principle between the NMZ 
formulation in the space of observables and associated space of 
states which extends the well-known duality between 
Koopman and Perron-Frobenious operators to reduced 
observable algebras and states.
We also provided guidance on the selection of the 
projection operators appearing in NMZ by proving that 
the only projections onto $C^*$-subalgebras that preserve 
all states are the conditional expectations -- a  
special class of projections on $C^{*}$-algebras.
Such projections can be determined 
systematically for a broad class of bounded and unbounded 
observables. This allows us to derive formally exact NMZ equations for 
observables and states in high-dimensional classical and 
quantum systems. Computing the solution to such equations
is usually a very challenging task that 
needs to address approximation of memory integrals and noise 
terms for which suitable (typically problem-class-dependent) 
algorithms are needed.

\begin{acknowledgments}
This work was supported by the Air Force Office of Scientific
Research grant FA9550-16-1-0092.
\end{acknowledgments}

\appendix
\section{Nondegenerate Homomorphisms and Approximate Identities}\label{app:nonDegenHomApproxIds}
\begin{definition}[Approximate Identity]
	Given a $C^{*}$-algebra $\mathfrak{A}$, a net $\{E_{\alpha}\}\subset\mathfrak{A}$ is an \emph{approximate identity} for $\mathfrak{A}$ if $E_{\alpha}\geq 0$ and $\|E_{\alpha}\|\leq 1$ for all $\alpha$ and if $E_{\alpha}A\to A$ for all $A\in\mathfrak{A}$.
\end{definition}

\begin{lemma}\label{lem:NondegenApproxIds}
Let $\mathfrak{A}$, $\mathfrak{B}$ be $C^{*}$-algebras and $\Psi:\mathfrak{B}\to\mathfrak{A}$ a $C^{*}$-homomorphism.  Then $\Psi$ is nondegenerate (i.e., $\Span_{\mathbb{C}}\{\Psi(b)a\;:\;b\in\mathfrak{B},a\in\mathfrak{A}\}$) if and only if $\Psi$ is \emph{approximately unital} [i.e., for some (and therefore every) approximate identity \cite{Segal1947a} $\{E_{\beta}\}\subset\mathfrak{B}$, $\{\Psi(E_{\beta})\}$ is a approximate identity for $\mathfrak{A}$].
\end{lemma}
\begin{proof}
First, assume that $\Psi$ is nondegenerate and let $\{E_{\beta}\}\subset\mathfrak{B}$ be an approximate identity.  For any $b\in\mathfrak{B}$, $\lim_{\beta}E_{\beta}b = b$, and by the continuity of $\Psi$, $\lim_{\beta}\Psi(E_{\beta}b) = \Psi(b)$.  Then for any $a\in\mathfrak{A}$, $\lim_{\beta}\Psi(E_{\beta}b)a = \Psi(b)a$.  Therefore
\begin{equation}
\lim_{\beta}\|\Psi(E_{\beta})\Psi(b)a - \Psi(b)a\|  = \lim_{\beta}\|\Psi(E_{\beta}b)a - \Psi(b)a\| = 0
\end{equation}
for any $b\in\mathfrak{B}$ and $a\in\mathfrak{A}$.  Since nondegeneracy of $\Psi$ implies that $\Span_{\mathbb{C}}\{\Psi(b)a\,:\,b\in\mathfrak{B}, a\in\mathfrak{A}\}$ is dense in $\mathfrak{A}$, we have found that $\Psi(E_{\beta})a\to a$ for all $a$ in a dense subspace of $\mathfrak{A}$.  Since $\|E_{\beta}\|\leq 1$ and therefore $\|\Psi(E_{\beta})\|\leq 1$, we conclude that $\lim_{\beta}\Psi(E_{\beta})a\to a$ for all $a\in\mathfrak{A}$, i.e. $\{\Psi(E_{\beta})\}$ is an approximate identity on $\mathfrak{A}$.

Now, suppose that $\Psi$ is degenerate, so that $T:=\Span_{\mathbb{C}}\{\Psi(b)a\;:\;b\in\mathfrak{B},a\in\mathfrak{A}\}$ is not dense in $\mathfrak{A}$.  Then there exist $\epsilon>0$ and $a\in\mathfrak{A}$ such that $\|a-t\|>\epsilon$ for all $t\in T$.  Then let $\{E_{\beta}\}$ be any approximate identity in $\mathfrak{B}$.  Since $\Psi(E_{\beta})a\in T$ for all $\beta$, $\|\Psi(E_{\beta})a - a\|>\epsilon$ for all $\beta$, and therefore $\Psi(E_{\beta})a\not\to a$, so that $\{\Psi(E_{\beta})\}$ is not an approximate identity.  So, by contradiction, if $\{\Psi(E_{\beta})\}$ is an approximate identity for some approximate identity $\{E_{\beta}\}$, then $\Psi$ must be nondegenerate.
\end{proof}

It may be noted that, if $\mathfrak{B}$ is unital, then $\Psi:\mathfrak{B}\to\mathfrak{A}$ is nondegenerate if and only if $\mathfrak{A}$ is a unital algebra and $\Psi$ is a unital $C^{*}$-homomorphism.  This follows from the simple fact that $E_{\beta}\equiv\identity$ is the only possible \emph{constant} approximate identity.

%\red{Thus nondegenerate $C^{*}$-homomorphisms are analogous to unital morphisms.  In the case of unital algebras, a non-unital morphism $\Psi:\mathfrak{B}\to\mathfrak{A}$ must map $\identity$ to a self-adjoint idempotent $P$, which acts as the identity in the subalgebra $P\mathfrak{A}P$.  Can we safely ignore this possibility for now and work solely in the category of $C^{*}$-algebras and nondegenerate $C^{*}$-homomorphisms?  What is the analogous form for non-unital algebras?}

\begin{lemma}
For any (contractive) approximate identity ${E_{\alpha}}\subset\mathfrak{A}$, $\|E_{\alpha}\|\to 1$.
\end{lemma}
\begin{proof}
For any nonzero $a\in\mathfrak{A}$, 
\begin{equation}
\liminf_{\alpha}\|E_{\alpha}\|\geq \liminf_{\alpha}\frac{\|E_{\alpha}a\|}{\|a\|} = \frac{\|a\|}{\|a\|} = 1,
\end{equation}
and, since $\|E_{\alpha}\|\leq 1$ for all $\alpha$, $\limsup_{\alpha}\|E_{\alpha}\|\leq 1$.  Therefore $\lim_{\alpha}\|E_{\alpha}\| = 1$.
\end{proof}

\begin{lemma}\label{lem:nonDegenNorm1}
Let $\mathfrak{A}$, $\mathfrak{B}$ be $C^{*}$-algebras and $\Psi:\mathfrak{B}\to\mathfrak{A}$ a nondegenerate $C^{*}$-homomorphism.  Then $\|\Psi\| = 1$.
\end{lemma}
\begin{proof}
	Since $\Psi$ is a $C^{*}$-homomorphism, $\|\Psi\|\leq 1$.  Let $\{E_{\beta}\}\subset \mathfrak{B}$ be an approximate identity.  Then $\{\Psi(E_{\beta})\}$ is also an approximate identity, and $\|E_{\beta}\|\to 1$ and $\|\Psi(E_{\beta})\|\to 1$, so that 
	\begin{equation}
		\|\Psi\|\geq\lim_{\beta}\frac{\|\Psi(E_{\beta})\|}{\|E_{\beta}\|} = 1,
	\end{equation}
	and therefore $\|\Psi\|=1$.
\end{proof}

\begin{lemma}\label{lem:nonDegenStatePreserv}
Let $\mathfrak{A}$, $\mathfrak{B}$ be $C^{*}$-algebras and $\Psi:\mathfrak{B}\to\mathfrak{A}$ a nondegenerate $C^{*}$-homomorphism. $\|\BD{\Psi}\phi\| = \|\phi\|$, for any $\phi\geq 0$, where $\BD{\Psi}:\BD{\mathfrak{A}}\to\BD{\mathfrak{B}}$ is the adjoint operator.
\end{lemma}
\begin{proof}
For any $\phi\in\mathfrak{A}^{*}$, $\phi\geq 0$, $\|\phi\| = \lim_{\beta}|\phi(F_{\beta})|$, where $\{F_{\beta}\}\subset\mathfrak{A}$ is an approximate identity.  Thus, for any approximate identity $\{E_{\beta}\}\subset \mathfrak{B}$, 
\begin{equation}
\|\BD{\Psi}\phi\| = \lim_{\beta}|(\BD{\Psi}\phi)(E_{\beta})| = \lim_{\beta}|\phi(\Psi(E_{\beta}))| = \|\phi\|
\end{equation}
since, by Lemma \ref{lem:NondegenApproxIds}, $\{\Psi(E_{\beta})\}$ is an approximate identity for $\mathfrak{A}$.
\end{proof}

\section{State-Preserving Maps}
\label{app:statePreservingMaps}
\begin{theorem}
\label{thm:statePreservingMaps}
Let $\mathfrak{A}, \mathfrak{B}$ be $C^{*}$-algebras, 
and $\Psi:\mathfrak{A}\to\mathfrak{B}$ a linear map 
satisfying $\BD{\Psi}[\mathcal{S}(\mathfrak{B})] \subset\mathcal{S}(\mathfrak{A})$.  Then $\Psi$ is a positive contraction with $\|\Psi\| = 1$.  If $\mathfrak{A}$ is unital, then $\mathfrak{B}$ is unital and $\Psi(\identity) =\identity$.
\end{theorem}
\begin{proof}
First note that $\BD{\Psi}[\mathcal{S}(\mathfrak{B})] \subset\mathcal{S}(\mathfrak{A})$ requires that $\BD{\Psi}$ be positive and  $\|\BD{\Psi}\phi\| = \|\phi\|$ for all $\phi\geq 0$, and by [\onlinecite[Th. 4.3.4]{Kadison1997}], $\BD{\Psi} \geq 0$ if and only if $\Psi \geq 0$:
\begin{align}
\BD{\Psi}\phi\geq 0 \,\forall\, \phi\geq 0 & \Leftrightarrow \BD{\Psi}\phi(g)\geq 0 \,\forall\, \phi\geq 0,\, g\geq 0 \notag\\
& \Leftrightarrow \phi(\Psi g)\geq 0 \,\forall\, \phi\geq 0,\, g\geq 0 \notag\\
& \Leftrightarrow \Psi g\geq 0 \,\forall\, g\geq 0. \notag
\end{align}
Now we pass to the second dual $\BDD{\mathfrak{A}}$ of $\mathfrak{A}$ which, via the Takeda-Sherman theorem \cite{Takeda1954,Blackadar2006} may be endowed with a multiplication which renders it a (unital) von Neumann algebra (the universal enveloping von Neumann algebra of $\mathfrak{A}$).  We likewise endow $\BDD{\mathfrak{B}}$ with the structure of a (unital) von Neumann algebra.  %The isometric isometries \cite[III.5.2.6]{Blackadar2006} between $\PD{\mathfrak{M}}$ and $\BD{\mathfrak{A}}$ and between $\PD{\mathfrak{N}}$ and $\BD{\mathfrak{B}}$ identify the bounded functionals of $\mathfrak{A}$ and $\mathfrak{B}$ with the normal functionals of $\mathfrak{M}$ and $\mathfrak{N}$,  allow $\BD{\Psi}:\BD{\mathfrak{B}}\to\BD{\mathfrak{A}}$ to be identified isometrically with a normal linear map $\Upsilon: \PD{\mathfrak{N}}\to \PD{\mathfrak{M}}$.  It inherits positivity from $\Psi$.  In addition, for each state $\phi\in\mathcal{S}(\mathfrak{B})$, passing to the identified normal state $\tilde{\phi}$ of $\mathfrak{N}$, we have $\tilde{\phi}(\identity_{\mathfrak{N}} - \BD{\Upsilon}(\identity_{\mathfrak{M}})) = \|\tilde{\phi}\| - \|\Upsilon\tilde{\phi}\| = \|\phi\| - \|\BD{\Psi}\phi\| = 0$ by assumption about $\Psi$ and \cite[Prop. II.6.2.5]{Blackadar2006}. 
 Then for each state $\phi\in\mathcal{S}(\mathfrak{B})$, we have $\phi(\identity_{\BDD{\mathfrak{B}}} - \BDD{\Psi}(\identity_{\BDD{\mathfrak{A}}})) = \phi(\identity_{\BDD{\mathfrak{B}}}) - \BD{\Psi}\phi(\identity_{\BDD{\mathfrak{A}}}) = \|\phi\| - \|\BD{\Psi}\phi\| = 0$ by assumption about $\Psi$ and [\onlinecite[Prop. II.6.2.5]{Blackadar2006}].  Since this holds for all states of $\mathfrak{B}$, which comprise all normal states of $\BDD{\mathfrak{B}}$, and they separate points in $\BDD{\mathfrak{B}}$, it follows that $\BDD{\Psi}(\identity_{\BDD{\mathfrak{A}}}) = \identity_{\BDD{\mathfrak{A}}}$, i.e. $\BDD{\Psi}$ is a unital positive map, and therefore is a contraction \cite{Russo1966} with $\|\BDD{\Psi}\| = 1$.  And because $\|\Psi\| = \|\BD{\Psi}\| = \|\BDD{\Psi}\|$, $\Psi$ is also a positive contraction with $\|\Psi\| = 1$.
\end{proof}

\begin{corollary}
\label{cor:statePreservingProjection}
  If $\mathfrak{A}$ is a $C^{*}$-algebra, $\Prj$ a linear projection on $\mathfrak{A}$ satisfying $\BD{\Prj}[\mathcal{S}(\mathfrak{A})]\subset\mathcal{S}(\mathfrak{A})$, and the image of $\Prj$ is a $C^{*}$-subalgebra $\mathfrak{B}\subset\mathfrak{A}$,  then $\Prj$ is a conditional expectation.
\end{corollary}
\begin{proof}
By \ref{thm:statePreservingMaps}, $\Prj$ is a contractive projection onto a $C^{*}$-subalgebra, and therefore is a conditional expectation \cite{Tomiyama1957,Blackadar2006}. 
\end{proof}

%\bigskip
%\section{Alternate Dyson Identity}
%The Dyson identity can be arranged differently, yielding a different expression for $\frac{\rmd }{\rmd t}\EO(t,t_{0})\Prj$.
%\begin{align}
%	Y(t,t_{0})-Z(t,t_{0}) &= \int_{t_{0}}^{t}\frac{d}{ds}\Big(-Z(s,0)Y(t,s)\Big)\,\rmd s \notag\\
%	& = \int_{t_{0}}^{t}Z(s,0)(A(s)-B(s))Y(t,s)\,\rmd s
%\end{align}
%\begin{align}
%\EO(t,t_{0})  & = \TOR e^{\int_{t_{0}}^{t}\PrjC\LV[\tau]\,\rmd\tau} + \int_{t_{0}}^{t}\TOR e^{\int_{t_{0}}^{s}\PrjC\LV[\tau]\,\rmd\tau}\Prj\LV[s]\EO(t,s)\,\rmd s,
%\end{align}
%\begin{subequations}
%\begin{align}
%\frac{\rmd}{\rmd t}\EO(t,t_{0})\Prj & = \EO(t,t_{0})\LV[t]\Prj\\
%& = \TOR e^{\int_{t_{0}}^{t}\PrjC\LV[\tau]\,\rmd\tau}\LV[t]\Prj \notag \\
% & \quad + \int_{t_{0}}^{t}\TOR e^{\int_{t_{0}}^{s}\PrjC\LV[\tau]\,\rmd\tau}\Prj\LV[s]\EO(t,s)\LV[t]\Prj\,\rmd
%  s\label{eqn:AltHeisenbergGLE}\\
%   & = \TOR e^{\int_{t_{0}}^{t}\PrjC\LV[\tau]\,\rmd\tau}\PrjC\LV[t]\Prj + \EO(t,t_{0})\Prj\LV[t]\Prj\notag \\
%   & \quad + \int_{t_{0}}^{t}\TOR e^{\int_{t_{0}}^{s}\PrjC\LV[\tau]\,\rmd\tau}\Prj\LV[s]\EO(t,s)\PrjC\LV[t]\Prj\,\rmd s  
%\end{align}
%\end{subequations}
%\red{[Is it fair to call this a generalized Langevin equation for $\EO(t,t_{0})\Prj$?  Are these expressions of any use (in particular, because they move the $\EO(t,s)$ to the right in the memory kernel)?]}

%\bibliographystyle{alphaurl}
%\bibliographystyle{abbrvnat}
%\bibliographystyle{apsrev4-1}
%\bibliography{refs}

%

\end{document}